\documentclass[journal,draftcls,onecolumn,11pt,twoside]{IEEEtran}

\usepackage[utf8]{inputenc} 
\usepackage[T1]{fontenc}% optional T1 font encoding

\ifCLASSINFOpdf

\else

\fi
\usepackage[cmex10]{amsmath}
\usepackage{ifthen}
\usepackage{tikz}
\usepackage{cite}
\usepackage{amsthm,amsfonts,amssymb}
\usepackage{graphicx}
\usepackage{subfig}
\usepackage{blkarray}
\usepackage{dsfont}
\usepackage[mathscr]{euscript}
\usepackage{enumitem}
\usepackage{algorithm}
\usepackage[noend]{algpseudocode}
\usepackage{blkarray}
\usepackage{stfloats}%
\newtheorem{theorem}{Theorem}[section]
\newtheorem{corollary}{Corollary}[section]
\newtheorem{proposition}{Proposition}[section]
\newtheorem{lemma}{Lemma}[section]
\theoremstyle{remark}

\theoremstyle{definition}
\newtheorem{definition}{Definition}[section]

\usepackage{url}
\usepackage{tcolorbox}
\usepackage{arydshln}
\usepackage{mathtools}
\usepackage{dsfont}

\definecolor{brickred}{cmyk}{0,0.89,0.94,0.28}
\definecolor{goldenrod}{cmyk}{0,0.10,0.84,0}
\definecolor{purple}{cmyk}{0.45,0.86,0,0}
\definecolor{rawsienna}{cmyk}{0,0.72,1,0.45}
\definecolor{olivegreen}{cmyk}{0.64,0,0.95,0.40}
\definecolor{peach}{cmyk}{0,0.5,0.7,0}
\definecolor{darkolive}{rgb}{0.,0.4,0.}
\colorlet{grey}{gray!40}

\global\long\def\E{\mathbb{E}}

\DeclareMathOperator*{\argmax}{arg\,max} \DeclareMathOperator*{\argmin}{arg\,min}

\usepackage{graphicx}
\makeatletter
\newcommand{\ostar}{\mathbin{\mathpalette\make@circled\star}}
\newcommand{\make@circled}[2]{%
	\ooalign{$\m@th#1\smallbigcirc{#1}$\cr\hidewidth$\m@th#1#2$\hidewidth\cr}%
}
\newcommand{\smallbigcirc}[1]{%
	\vcenter{\hbox{\scalebox{0.77778}{$\m@th#1\bigcirc$}}}%
}
\makeatother
\usepackage[cmintegrals]{newtxmath}

\newcommand{\gbinom}{\genfrac{[}{]}{0pt}{}}

\begin{document}

\title{An Analysis of RPA Decoding of Reed-Muller Codes Over the BSC}
\author{
%\IEEEauthorblockN{V.~Arvind~Rameshwar}
%\IEEEauthorblockA{IUDX Program Unit\\Indian Institute of Science, Bengaluru, India\\
%	Email: arvind.rameshwar@gmail.com}
%\and
%\IEEEauthorblockN{V.~Lalitha}
%\IEEEauthorblockA{Signal Processing and Communications Research Center (SPCRC)\\ IIIT Hyderabad, India\\
%	Email: lalitha.v@iiit.ac.in}
\IEEEauthorblockN{V.~Arvind~Rameshwar,~\IEEEmembership{Member,~IEEE}  and
	V.~Lalitha,~\IEEEmembership{Senior Member,~IEEE}}
\thanks{V.~A.~Rameshwar is with the India Urban Data Exchange Program Unit, SID, Indian Institute of Science, Bengaluru, India,	email: \texttt{arvind.rameshwar@gmail.com}. V.~Lalitha is with the Signal Processing and Communications Research Center (SPCRC), International Institute of Information Technology, Hyderabad, India, email: \texttt{lalitha.v@iiit.ac.in}. A part of this work has been accepted to the 2025 IEEE International Symposium on Information Theory (ISIT).}
}
\IEEEoverridecommandlockouts
\markboth{}%
{Rameshwar and Lalitha: An Analysis of RPA Decoding Over the BSC}

\maketitle

\begin{abstract}
	In this paper, we revisit the Recursive Projection-Aggregation (RPA) decoder, of Ye and Abbe (2020), for Reed-Muller (RM) codes. Our main contribution is an explicit upper bound on the probability of incorrect decoding, using the RPA decoder, over a binary symmetric channel (BSC). Importantly, we focus on the events where a \emph{single} iteration of the RPA decoder, in each recursive call, is sufficient for convergence. Key components of our analysis are explicit estimates of the probability of incorrect decoding of first-order RM codes using a maximum likelihood (ML) decoder, and estimates of the error probabilities during the aggregation phase of the RPA decoder. Our results allow us to show that for RM codes with blocklength $N = 2^m$, the RPA decoder can achieve  vanishing error probabilities, in the large blocklength limit, for RM orders that grow roughly logarithmically in $m$.
\end{abstract}

\IEEEpeerreviewmaketitle
\section{Introduction}
\label{sec:introduction}
Reed-Muller (RM) codes are a family of binary linear codes that are obtained by the evaluations of Boolean polynomials on the points of the Boolean hypercube \cite{reed,muller}. They are of widespread interest in the information theory and error-control coding communities, owing to recent breakthrough theoretical progress has shown that RM codes by themselves are  capacity-achieving for the binary erasure channel \cite{kud1}, and more generally, for BMS channels \cite{Reeves,abbesandon}. Further, they share a close relationship \cite{arikanrmpolar-1,arikanrmpolar-2} with polar codes \cite{polar}---the first explicit family of (linear) codes that was shown to achieve capacity over binary-input memoryless symmetric (BMS) channels and also provably achieve high rates over many channels with memory \cite{shuvaltal,sasoglutal1}.

In the light of their importance, much work has been dedicated to finding good decoding algorithms for RM codes. The earliest such algorithm by Reed \cite{reed} is capable of correcting bit-flip errors up to half the minimum distance of the RM code. For the special case of first-order RM codes, a Fast Hadamard Transform-based (or FHT-based) decoder was designed in \cite{fht2,fht}, which is an efficient implementation of a maximum likelihood (ML) decoding procedure. For second-order RM codes, a decoding algorithm that works well for moderate blocklength $N$ ($N\leq 1024$) RM codes was discussed in \cite{sidelnikov,sakkour}. Further, recursive decoding algorithms for RM codes (with and without lists) that display good performance at moderate blocklengths were described and analyzed in \cite{dumer1,dumer2,dumer3,dumer_burnashev}. We also refer the reader to several recent data-driven and heuristic approaches for decoding RM codes \cite{ko,jamali,redundant1,redundant2}.

A much more recent decoding algorithm, variants of which were shown to achieve near-ML performance at moderate blocklengths over the binary symmetric channel (BSC), is the Recursive Projection-Aggregation (RPA) decoder of Ye and Abbe \cite{rpa}. Via extensive simulation studies, the work \cite{rpa} showed that the RPA algorithm outperforms the decoding algorithms in \cite{sidelnikov,sakkour} for second-order RM codes and the decoding algorithms in \cite{dumer1,dumer2,dumer3}, for selected moderate blocklength RM codes. Later works \cite{rpacomplex1,rpacomplex2} presented procedures for reducing the complexity of the RPA algorithm. 
%We also refer the reader to recent work \cite{coseterror} on using only a subset of the total number of subspaces of a fixed dimension could lead to poor performance, due to the decoder's potential failure in correcting so-called ``coset error patterns".

While it is now well-accepted that the RPA algorithm (as described in \cite{rpa}) demonstrates good performance in practice for low-rate (in particular, fixed-order) RM codes, a formal, theoretical analysis of the RPA decoder has been largely lacking in the literature. In this paper, we make progress towards this goal, by obtaining explicit upper bounds on the error probabilities of the RPA decoder, when the channel used is the BSC. Importantly, we prove that for sufficiently large blocklengths, with overwhelmingly high probability, the RPA decoder recovers the correct input codeword in RM codes of order that scales up to roughly logarithmically in the parameter $m$, where $N = 2^m$. This result can hence be viewed as supplementing the results in \cite{asw,sberlo}, where it was shown that RM codes of order growing linearly in the parameter $m$, but with vanishing rate, have error probabilities that decay to zero in the large blocklength limit. Our result, although weaker in terms of the orders of the RM codes of interest, establishes that an \emph{explicit, low-complexity} decoding algorithm (namely, the RPA algorithm) can in fact be used to achieve vanishing error probabilities, for Reed-Muller orders that grow with $m$.

Our analysis proceeds via first obtaining explicit upper bounds on the probability of incorrect maximum likelihood (ML) decoding of first-order RM codes, which is used as a subroutine in the RPA decoder. With such upper bounds in place, we then characterize upper bounds on the probabilities of error in the ``aggregation" steps carried out by the RPA decoder, conditioned on the event that all the ``projections" are decoded correctly. Interestingly, our results show that even when one restricts attention to the event that each of the recursive calls to the RPA decoder converges in just \textit{one} iteration, it is possible to obtain non-trivial performance estimates.

Our results can be seen as the first in a line of work to characterize the \emph{largest} order of an RM code, as a function of the parameter $m$, which can be decoded via RPA decoding with vanishing error probability, in the limit as the blocklength increases to infinity. Such analysis, however, will require a finer understanding of the correlations among the random variables that are recursed on over \emph{several} (more than one) iterations of the RPA decoder. Importantly, via such more sophisticated analysis, we shall be able to determine if RM codes can achieve positive rates over (and potentially the capacity of) the BSC under RPA decoding.
\section{Notation and Preliminaries}

\subsection{Notation}
\label{sec:notation}
For a positive integer $n$, we use $[n]$ as shorthand for $\{1,2,\ldots,n\}$. Given a set $\mathcal{X}$, we use the notation $\mathds{1}_\mathcal{X}(x)$ to be the indicator function that equals $1$ when $x\in \mathcal{X}$, and equals $0$, otherwise.
Random variables are denoted by capital letters, e.g., $X, Y$, and small letters, e.g., $x, y$, denote their instantiations. %Sets are denoted by calligraphic letters, e.g., $\mathcal{X}, \mathcal{Y}$. 
The notation $\mathbf{0}$ denotes the all-zeros vector, whose length can be inferred from the context.
%Notation such as $P(x), P(y|x)$ are used to denote the probabilities $P_X(x), P_{Y|X}(y|x)$, when it is clear which random variables are being referred to. The notations ${H}(X) :=  \mathbb{E}[-\log P(X)], {H}(Y\mid X):= \mathbb{E}[-\log P(Y\mid X)]$, and $I(X;Y) := H(Y)-H(Y\mid X)$ denote the entropy of $X$, conditional entropy of $Y$ given $X$, and mutual information between $X$ and $Y$, respectively. Given any real $p\in [0,1]$, we let $h_b(p):=-p\log p -(1-p)\log(1-p)$, where $h_b$ denotes the binary entropy function; here, the base of the logarithm will be made clear from the context. 
The notation Ber$(p)$ refers to the Bernoulli distribution with parameter $p$, where $p\in [0,1]$. Likewise, we use the notation Ber$^{\otimes n}(p)$ to denote a length-$n$ random vector, each of whose coordinates is independent and identically distributed (i.i.d.) according to the Ber$(p)$ distribution. For a vector $\mathbf{a}\in \{0,1\}^n$, we denote by $\mathbf{a}^\pm$ as that vector that obeys $a^\pm_i = (-1)^{a_i}$, for each $i\in [n]$. We use the symbol `$\oplus$' to denote addition modulo $2$, i.e., for vectors $\textbf{a},\textbf{b}\in \{0,1\}^n$, for some $n\geq 1$ we have $\textbf{a}\oplus \mathbf{b} = \left(a_i+b_i\ (\text{mod $2$}):\ i\in [n]\right)$. For a vector $\mathbf{a}\in \{0,1\}^n$, we denote its Hamming weight by $w(\mathbf{a})$; for $\mathbf{u},\mathbf{v}\in \{0,1\}^n$, we let $d(\mathbf{u},\mathbf{v})$ denote the Hamming distance between $\mathbf{u}$ and $\mathbf{v}$.. Natural logarithms are denoted as $\ln$ and $e^x$, for $x\in \mathbb{R}$, is sometimes denoted as exp$(x)$.

\subsection{Reed-Muller Codes}
\label{sec:rmintro}
We now recall the definition of the binary Reed-Muller (RM) family of codes and some of their basic facts that are relevant to this work. Codewords of binary RM codes consist of the evaluation vectors of multivariate polynomials over the binary field $\mathbb{F}_2$. Consider the polynomial ring $\mathbb{F}_2[x_1,x_2,\ldots,x_m]$ in $m$ variables. Note that any polynomial $f\in \mathbb{F}_2[x_1,x_2,\ldots,x_m]$ can be expressed as the sum of {monomials} of the form $\prod_{j\in S:S\subseteq [m]} x_j$, since $x^2 = x$ over the field $\mathbb{F}_2$. For a polynomial $f\in \mathbb{F}_2[x_1,x_2,\ldots,x_m]$ and a binary vector $\mathbf{z} = (z_1,\ldots,z_m)\in \mathbb{F}_2^m$, we write {$f(\mathbf{z})=f(z_1,\ldots,z_m)$} as the evaluation of $f$ at $\mathbf{z}$. The evaluation points are ordered according to the standard lexicographic order on strings in $\mathbb{F}_2^m$, i.e., if $\mathbf{z} = (z_1,\ldots,z_m)$ and $\mathbf{z}^{\prime} = (z_1^{\prime},\ldots,z_m^{\prime})$ are two evaluation points, then, $\mathbf{z}$ occurs before $\mathbf{z}^{\prime}$ iff for some $i\geq 1$, we have $z_j = z_j^{\prime}$ for all $j<i$, and $z_i < z_i^{\prime}$. Now, let Eval$(f):=\left({f(\mathbf{z})}:\mathbf{z}\in \mathbb{F}_2^m\right)$ be the evaluation vector of $f$, where the coordinates $\mathbf{z}$ are ordered according to the standard lexicographic order. 

\begin{definition}[see Ch. 13 in \cite{mws}, or \cite{rm_survey}]
	{For $0\leq r\leq m$}, the $r^{\text{th}}$-order binary Reed-Muller code RM$(m,r)$ is defined as
	\[
	\text{RM}(m,r):=\{\text{Eval}(f): f\in \mathbb{F}_2[x_1,x_2,\ldots,x_m],\ \text{deg}(f)\leq r\},
	\]
	where $\text{deg}(f)$ is the degree of the largest monomial in $f$, and the degree of a monomial $\prod_{j\in S: S\subseteq [m]} x_j$ is simply $|S|$. 
\end{definition}

It is known that the evaluation vectors of all the distinct monomials in the variables $x_1,\ldots, x_m$ are linearly independent over $\mathbb{F}_2$. Hence, RM$(m,r)$ has dimension $\binom{m}{\le r} := \sum_{i=0}^{r}{m \choose i}$. We then have that RM$(m,r)$ is a $\left[2^m,\binom{m}{\le r}\right]$ linear code. It is also known that RM$(m,r)$ has minimum Hamming distance $d:={d}_{\text{min}}(\text{RM}(m,r))=2^{m-r}$.

%{Furthermore, the dual code of RM$(m,r)$ is RM$(m,m-r-1)$, for $m\geq 1$ and $r\leq m-1$.}  Of importance in this paper is the fact that each minimum-weight codeword of RM$(m,r)$ is the characteristic vector of an $(m-r)$-dimensional affine subspace of $\mathbb{F}_2^m$ (see \cite[Thm. 8, Ch. 13, p. 380]{mws}) , i.e., the vector having ones in those coordinates $\mathbf{z}\in \{0,1\}^m$ that lie in the affine subspace. We use this fact to efficiently sample  a uniformly random minimum-weight codeword of a given RM code. Another fact of use to us is that the collection of minimum-weight codewords spans RM$(m,r)$, for any $m\geq 1$ and $r\leq m$ (see \cite[Thm. 12, Ch. 13, p. 385]{mws}).
\subsection{System Model and the RPA Decoder}
\label{sec:rpadef}
In this subsection, we set up the system model of interest and briefly recall the Recursive Projection-Aggregation (RPA) algorithm introduced in \cite{rpa}, for this setting.
\subsubsection{System Model}
In this paper, we work with the setting of the transmission of a codeword of an RM code across a binary symmetric channel BSC$(p)$, where $p\in (0,\frac12)$ is the crossover probability of the channel. Let $\mathbf{c}$ be the codeword that is drawn from the code RM$(m,r)$, for $m\geq r\geq 1$ (the case $r=0$ corresponds to a simple repetition code), for transmission. We write $N:=2^m$ as the blocklength of the code. The codeword $\mathbf{c}$ is transmitted across a BSC$(p)$, resulting in the (random) received vector $\mathbf{Y}$, where
\[
\mathbf{Y} = \mathbf{c}\oplus \boldsymbol{\nu},
\]
where $\boldsymbol{\nu}\sim \text{Ber}^{\otimes N}(p)$. We index the coordinates of $\mathbf{Y}$ using vectors $\mathbf{z}\in \{0,1\}^m$ in lexicographic order. The RPA decoder will use this received vector $\mathbf{Y}$ to obtain an estimate of the transmitted codeword $\mathbf{c}$.
\subsubsection{The RPA Decoder}
Fix a positive integer $k\leq r-1$, which denotes the dimension of the subspaces used in the RPA algorithm. We assume throughout this paper that $k$ divides $r-1$ (denoted as $k|(r-1)$). For any $k$-dimensional subspace $\mathbb{B}$ of $\{0,1\}^m$, $i\in [N-1]$, that contains the non-zero vector $\mathbf{b}$, we let $\{0,1\}^m / \mathbb{B}$ denote the collection of cosets of $\mathbb{B}$ in $\{0,1\}^m$. Now, consider the collection $\{\mathbb{B}_i\subseteq \{0,1\}^m\}$ of $k$-dimensional subspaces of $\{0,1\}^m$. The total number of such subspaces $n_{k,m}$ is given by the following Gaussian binomial coefficient \cite[Ch. 15, Sec. 2, Thm. 9]{mws}:
\begin{equation}
	\label{eq:gaussbinom}
	n_{k,m} = \gbinom{m}{k}:= \prod_{i=0}^{k-1}\frac{2^m-2^i}{2^k-2^i}.
\end{equation}
Note that when $k=1$, we have $n_{1,m} = 2^m-1 = N-1$. 

Given a set $\mathcal{S}\subseteq \{0,1\}^m$, we use the notation $\bigoplus_{\mathbf{x}\in \mathcal{S}} Y_\mathbf{x}$ to denote the modulo-$2$ sum of all coordinates of $\mathbf{Y}$ whose indices lie in $\mathcal{S}$. For any $i\in [n_{k,m}]$, any coset $T\in \{0,1\}^m / \mathbb{B}_i$ that contains the vector $\mathbf{x}$ (we denote such a coset as $[\mathbf{x}+\mathbb{B}_i]$), and for the given received vector $\mathbf{Y}$, we define ${Y}_{/\mathbb{B}_i}(T):= \bigoplus_{\mathbf{b}\in \mathbb{B}_i} Y_{\mathbf{x}\oplus\mathbf{b}}$ to be the ``projection" of $\mathbf{Y}$ onto the coset $T\in \{0,1\}^m/\mathbb{B}_i$. Finally, we define the projection of $\mathbf{Y}$ onto the cosets of $\mathbb{B}_i$ as
\begin{equation}
	\label{eq:project}
\mathbf{Y}_{/\mathbb{B}_i}:=\left({Y}_{/\mathbb{B}_i}(T):\ T\in \{0,1\}^m/\mathbb{B}_i\right),
\end{equation}
for some fixed ordering among cosets $T$. One can define a projection $\mathbf{c}_{/\mathbb{B}_i}$ analogously for the input codeword $\mathbf{c}$ too; via standard properties of Reed-Muller codes, it can be argued that $\mathbf{c}_{/\mathbb{B}_i}\in \text{RM}(m-k,r-k)$, if $\mathbf{c}\in \text{RM}(m,r)$. We now briefly recall the RPA decoder for second-order RM codes. The main decoding algorithm is shown as Algorithm \ref{alg:rpa}, which in turn uses the ``aggregation" step in Algorithm \ref{alg:agg} as a subroutine. We observe that the RPA decoding procedure consists of two key steps: (i) decoding of first-order ($r=1$) RM codes via the Fast Hadamard Transform (FHT) decoder (see \cite{fht} and \cite[Sec. V.B.1]{rm_survey}), which is simply an efficient implementation of a maximum likelihood (ML) decoder for first-order RM codes, and (ii) an aggregation step that aggregates multiple noisy estimates of a single codeword symbol. 

%We mention that our presentation of the RPA decoding algorithm and the aggregation subroutine are specialized to the setting where one projects of the received vector $\mathbf{Y}$ using \emph{one-dimensional subspaces}, as against the more general formulation in \cite{rpa}, which allows for the use of subspaces of higher dimension. As we shall see, projections via one-dimensional subspaces by themselves give rise to correct decoding, with very high probability, for sufficiently large blocklengths $N$.

\begin{algorithm}[t]
	\caption{The RPA decoding algorithm for RM$(m,r)$ using subspaces of dimension $k$}
	\label{alg:rpa}
	\begin{algorithmic}[1]	
		%		\State: Construct a one-one mapping $\phi_m: \mathbb{F}_2^{N_m-K_m}\to \mathcal{C}_{\log_2\left(\frac{N_m-K_m}{2^{-\left \lceil \log_2(d+1)\right \rceil}R}\right)}(R)$, where $\mathcal{C}_{\log_2\left(\frac{N_m-K_m}{2^{-\left \lceil \log_2(d+1)\right \rceil}R}\right)}(R)$ is constructed as in \eqref{eq:rmlb1}.
		\Procedure{\textsc{RPA}}{$\mathbf{Y}$, $m$, $r$, $M_\text{iter}$}
		\State Set $N\gets 2^m$.
		\For{$j = 1:M_\text{iter}$}
		\State Compute the projections $\mathbf{Y}_{/\mathbb{B}_i}$, $i\in [n_{k,m}]$.
		\If{$r\neq k+1$}
		\State For each $i\in [n_{k,m}]$, set $\widehat{\mathbf{Y}}_{/\mathbb{B}_i}\gets$ \textsc{RPA}($\mathbf{Y}_{/\mathbb{B}_i}, m-k, r-k, M_\text{iter}$). \label{step:recurse-high}
		\Else 
		\State For each $i\in [n_{k,m}]$, decode $\mathbf{Y}_{/\mathbb{B}_i}$ to $\widehat{\mathbf{Y}}_{/\mathbb{B}_i}$ using the FHT-based decoder. \label{step:fht-high}
		\EndIf
		\State Compute $\overline{\mathbf{Y}}\gets$ \textsc{Aggregate}$(\mathbf{Y},\widehat{\mathbf{Y}}_{/\mathbb{B}_1},\ldots,\widehat{\mathbf{Y}}_{/\mathbb{B}_{n_{k,m}}})$.\label{step:agg-high}%
		\If{$\mathbf{Y} = \overline{\mathbf{Y}}$}
		\textbf{break}
		\EndIf
		\State Update $\mathbf{Y}\gets \overline{\mathbf{Y}}$.
		\EndFor
		\State Return $\mathbf{Y}$.
		\EndProcedure	
	\end{algorithmic}
\end{algorithm} 

\begin{algorithm}[t]
	\caption{The aggregation subroutine}
	\label{alg:agg}
	\begin{algorithmic}[1]	
		%		\State: Construct a one-one mapping $\phi_m: \mathbb{F}_2^{N_m-K_m}\to \mathcal{C}_{\log_2\left(\frac{N_m-K_m}{2^{-\left \lceil \log_2(d+1)\right \rceil}R}\right)}(R)$, where $\mathcal{C}_{\log_2\left(\frac{N_m-K_m}{2^{-\left \lceil \log_2(d+1)\right \rceil}R}\right)}(R)$ is constructed as in \eqref{eq:rmlb1}.
		\Procedure{\textsc{Aggregate}}{$\mathbf{Y},\widehat{\mathbf{Y}}_{/\mathbb{B}_1},\ldots,\widehat{\mathbf{Y}}_{/\mathbb{B}_{n_{k,m}}}$}
		\State Initialize the length-$2^m$ vector $\textsc{Flip}\gets \mathbf{0}$.
		\For{$\mathbf{x}\in \{0,1\}^m$}
		\State Compute $\phi(\mathbf{x}) = \sum_{i=1}^{n_{k,m}} \mathds{1}\{Y_{/\mathbb{B}_i}([\mathbf{x}+\mathbb{B}_i])\neq \widehat{Y}_{/\mathbb{B}_i}([\mathbf{x}+\mathbb{B}_i])\}$.
		\State Set \textsf{Flip}$(\mathbf{x}) = 1$, if $\phi(\mathbf{x}) > \frac{n_{k,m}}{2}$.
		\EndFor
		\State Set $\overline{Y}\gets \mathbf{Y}\oplus \textsc{Flip}$.
		\State Return $\overline{Y}$.
		\EndProcedure	
	\end{algorithmic}
\end{algorithm} 
\section{Main Results}
\label{sec:main-result}
Given a fixed Reed-Muller code RM$(m,r)$, suppose that the (fixed) codeword $\mathbf{c}\in \text{RM}(m,r)$ is transmitted across the BSC$(p)$, and let $\widehat{\textbf{C}}$ denote the estimate produced by the RPA decoder. We then define the probability of error on transmission of $\mathbf{c}$ as $P_\text{err}(\mathbf{c}):=\Pr\left[\widehat{\mathbf{C}}\neq \mathbf{c}\mid \mathbf{c}\ \text{transmitted}\right]$ and the probability of error of the code RM$(m,r)$, under RPA decoding, when a codeword $\mathbf{C}$ is chosen uniformly at random from the code as $P_\text{err}(\text{RM}(m,r)):=\Pr\left[\widehat{\mathbf{C}}\neq \mathbf{C}\right]$.

An important fact about the error probabilities above is \cite[Prop. 2]{rpa}, which states that for all $m\geq 1, 1\leq r\leq m$, we have that the probability of error over the BSC$(p)$, upon transmission of a codeword $\mathbf{c}\in \text{RM}(m,r)$ is independent of $\mathbf{c}$. This, in turn, implies that we can focus our analysis on the case when $\mathbf{c} = \mathbf{0}\in \{0,1\}^N$, with $P_\text{err}(\text{RM}(m,r)) = P_\text{err}(\mathbf{0})$. In all the analysis that follows, we shall make the simplifying assumption that the input codeword is fixed to be the all-zeros codeword $\mathbf{0}$.

Before we state our main results, we will require some more notation. Let $\overline{p}:= \frac12\cdot (1-(1-2p)^{2^{r-2}})$ and for any $\alpha\in (0,\frac12)$, let $\eta(\alpha):= \frac{1}{2}\cdot (1-4\alpha(1-\alpha))$. Note that $\eta(\alpha)< \frac12$, for all $\alpha\in (0,\frac12)$. In what follows, we first focus on performance guarantees of RPA decoding when one-dimensional subspaces are used for projection. We then discuss some consequences of our main result on error probabailities using such subspaces and then state our result on error probabilities of RPA decoding using subspaces of dimension $k>1$. Our first main result is encapsulated in the following theorem:

\begin{theorem}
	\label{thm:rpaerror}
	For any $0<\epsilon<\eta(\overline{p})$, we have that for $r\geq 2$, using one-dimensional subspaces for projection,
	\[
	P_\text{\normalfont err}(\text{\normalfont RM}(m,r))\leq 32N^{r+1}\cdot \text{\normalfont exp}\left(-2^{-r-1}{N\epsilon^2}\right).
	\]
\end{theorem}
It can easily be checked that for any \emph{fixed} (constant) $r\geq 2$, we have that $$\lim_{m\to \infty} P_\text{\normalfont err}(\text{\normalfont RM}(m,r)) = \lim_{N\to \infty} P_\text{\normalfont err}(\text{\normalfont RM}(m,r)) = 0.$$

In fact, the following stronger corollary holds, by simply evaluating the limit of the right-hand side in Theorem \ref{thm:rpaerror}. Let $c= c(p) := \frac{\log 2}{\log\left(\frac{1}{1-2p}\right)}$. 
\begin{corollary}
	\label{cor:logr}
	For any $0<\overline{c}<c$, we have that for all $r\leq \log(\overline{c}m)$,
	\[
	\lim_{m\to \infty} P_\text{\normalfont err}(\text{\normalfont RM}(m,r))  = 0.
	\]
\end{corollary}
\begin{proof}
	First, observe that 
	\begin{align*}
		\eta(\overline{p}) &= \frac{1}{2}\cdot \left(1-4\overline{p}(1-\overline{p})\right)\\
		&= \frac{1}{2}\cdot\left(1-\left(1-(1-2p)^{2^{r-2}}\right)\cdot  \left(1+(1-2p)^{2^{r-2}}\right)\right)\\
		&= \frac12\cdot \left(1-2p\right)^{2^{r-1}}.
	\end{align*}
Now, from the upper bound in Theorem \ref{thm:rpaerror}, for any fixed $\delta\in (0,1)$, writing $\epsilon^2 = \delta\cdot \eta(\overline{p})^2$, we have $P_\text{err}(\text{RM}(m,r))\leq 2^{\rho_m(r,\delta)}$, where
\begin{align*}
	\rho_m(r,\delta):= m(r+1)-\frac{\delta\log_2(e)}{8}\cdot 2^{m-r}\cdot (1-2p)^{2^r}+5.
\end{align*}
By picking $r\leq \log(\overline{c}m)$, for any $\overline{c}<c$, it can be checked that the corollary indeed holds.
\end{proof}
Note, however, that RM codes with $r\leq \log(\overline{c}m)$, where $\overline{c}<c$, have rate $R_m:= \frac{\log |\text{RM}(m,r)|}{N} = \frac{{m\choose \leq r}}{N}$ such that $R_m \xrightarrow{m\to \infty} 0$. In other words, via the results in this paper, we only obtain that RPA decoding achieves vanishing error probabilities, in the limit as the blocklength goes to infinity, for selected RM codes that have \emph{vanishing rate}. An important question hence, as mentioned earlier, is whether RM codes can achieve \emph{positive rate} (or even the capacity of the BSC$(p)$) under RPA decoding---a fertile topic for future research.

Now, for \emph{fixed} (constant) orders $r$, Theorem \ref{thm:rpaerror} also allows us to understand how quickly, as a function of $m$, the crossover probability $p$ (also called the ``decoding radius") can converge to its maximum value that is $1/2$, while also ensuring that $P_\text{err}(\text{RM}(m,r))\xrightarrow{m\to \infty} 0$. Observe that the statement that a decoding algorithm can achieve vanishing error probabilities, in the limit as the blocklength $N$ tends to infinity, is equivalent to the statement that the algorithm can correct $N(p-\delta)$ errors with high probability, for any $\delta\in (0,1)$. The following result hence establishes a lower bound on the number of such errors correctable with high probability, as a function of $m$, for fixed orders $r$, in a vein similar to the results in \cite{dumer1}. Recall that $d = 2^{m-r}$ is the minimum distance of RM$(m,r)$.
\begin{corollary}
	\label{cor:decodingradius}
	For fixed $r\geq 2$ and any $\delta,\beta\in (0,1)$, for sufficiently large $m$, the RPA decoder can correct $\frac{N}{2}\cdot\left(1-\gamma_m-\delta\right)$ errors with high probability, where
	\[
	\gamma_m = \left(\frac{8(5+m(r+1+\beta))}{d\cdot \log_2 e}\right)^{\frac{1}{2^r}}.
	\]
\end{corollary}
\begin{proof}
	Fix $r\geq 2$. In order to ensure that $P_\text{err}(\text{RM}(m,r))\xrightarrow{m\to \infty} 0$, from Theorem \ref{thm:rpaerror}, we would like $p$ to be such that for some fixed $\beta\in (0,1)$,
	\begin{equation}
	e^{2^{-(r+1)}\cdot N\eta(\overline{p})^2}\geq 32N^{r+1+\beta}.
	\label{eq:tempradius}
	\end{equation}
	By using the fact that $\eta(\overline{p}) = \frac{1}{2}\cdot (1-2p)^{2^{r-1}}$ (see the proof of Corollary \ref{cor:logr}), we see via straightforward algebraic manipulations that \eqref{eq:tempradius} is equivalent to requiring that
	$
	p\leq \frac12\cdot \left(1-\gamma_m\right),
	$
	where $\gamma_m$ is as in the statement of the corollary. We hence obtain that for sufficiently large $m$, the RPA decoder can correct $N(p-\delta) = \frac{N}{2}\cdot\left(1-\gamma_m-\delta\right)$ errors with high probability, for any fixed $\delta\in (0,1)$.
\end{proof}
The result in Corollary \ref{cor:decodingradius} allows for a direct comparison with the main results in \cite{dumer1}. In particular, \cite[Thm. 2]{dumer1} shows that there exists a decoder for RM$(m,r)$, for fixed orders $r$, which can correct $\frac{N}{2}\cdot\left(1-\beta_m\right)$ errors with high probability, for sufficiently large $m$, where $\beta_m = \left(\frac{qm}{d}\right)^{\frac{1}{2^r}}$, for $q>\ln 4$. Note that for $r\geq 2$, we have
\[
\gamma_m \geq \left(\frac{24m}{d\cdot \log_2 e}\right)^{\frac{1}{2^r}}>\left(\frac{2m\cdot \ln 4}{d}\right)^{\frac{1}{2^r}}.
\]
This implies that the convergence of the decoding radius to $1/2$ (captured by $\gamma_m$), for fixed orders $r\geq 2$, under our analysis of RPA decoding, is slower than that obtainable using the decoding algorithm $\Phi_r^m$ in \cite{dumer1} (captured by $\beta_m$). However, for sufficiently large $m$, it can be verified similarly that the convergence of the decoding radius to $1/2$ under RPA decoding is faster than the decoding algorithm $\Psi_r^m$ in \cite{dumer1}. It is an interesting question to check whether the RPA decoder can outperform $\Phi_r^m$ in \cite{dumer1} in terms of the speed of convergence of the decoding radius, using more sophisticated analysis of the RPA error probabilities, possibly by considering cases where each RPA run converges after two or more iterations.

Our final result obtains an upper bound on the probability of error of RPA decoding using subspaces of dimension $k>1$ for projection. For any $j\in [r-1]$, let $p^{(j)}:= \frac12\cdot \left(1-(1-2p)^{2^{j}}\right)$; note that $\overline{p} = p^{(r-2)}$.
\begin{theorem}
	\label{thm:rpa-error-high}
	For any $0<\epsilon<\eta(\overline{p})$, we have that for $r\geq 2$, using $k$-dimensional subspaces for projection, where $k|(r-1)$,
	\begin{align*}
		&P_\text{\normalfont err}(\text{\normalfont RM}(m,r))
		\leq 64N^3\cdot n_{k,m}^{\frac{r-1}{k}}\cdot \text{\normalfont exp}\left(-\ln\left(\frac{1-p^{(r-k-1)}}{p^{(r-k-1)}}\right)\cdot 2^{-r-1-k}N\epsilon^2\right).
	\end{align*}
\end{theorem}
We argue in Section \ref{sec:higher-subspaces}, using Theorem \ref{thm:rpa-error-high}, that the guarantee on the growth rate of $r$ with $m$ to allow for vanishing error probabilities is the same as that in Corollary \ref{cor:logr} above (see Corollary \ref{cor:error-high} for details). Although one expects that the use of higher-dimensional subspaces for projection should lead to significantly improved estimates of error probabilities, our current analysis is unable to demonstrate such an improvement, mainly owing to the use of certain somewhat generic, albeit powerful, results to show the concentration of some random variables of interest about their expected value. We believe that strengthening this concentration inequality in the setting of higher-dimensional subspaces will lead to better bounds on the $P_\text{\normalfont err}(\text{\normalfont RM}(m,r))$---a task that we leave for future work.

\subsection{Organization of the Paper}
Sections \ref{sec:second} and \ref{sec:higher} are dedicated to the setting where one-dimensional subspaces are used for projection, and lead to the proof of Theorem \ref{thm:rpaerror}. Theorem \ref{thm:rpa-error-high} is then proved in Section \ref{sec:higher-subspaces}, using the tools introduced in the previous sections. 

In Section \ref{sec:second}, we shall first analyze the behaviour of the FHT decoder in Step \ref{step:fht-high} of Algorithm \ref{alg:rpa} (which acts on the first-order RM codes obtained via the projection step) and the aggregation step in Step \ref{step:agg-high} of Algorithm \ref{alg:rpa}, separately. Such analysis immediately leads to an understanding of the performance of the RPA decoder on \emph{second-order} ($r=2$) RM codes. In Section \ref{sec:higher}, we shall use the results obtained on the error probabilities of the FHT decoder and in the aggregation step, recursively, to get estimates of the error probabilities of the RPA decoder, using one-dimensional subspaces, for general orders $r>2$. With the recursive analysis of error probabilities in place in Section \ref{sec:higher}, we then briefly discuss the modifications that need to be made to extend the analysis to the setting where higher-dimensional subspaces are used in the RPA algorithm, in Section \ref{sec:higher-subspaces}, thus leading to a proof of Theorem \ref{thm:rpa-error-high}.

\section{Analysis on Second-Order RM Codes}
\label{sec:second}
In this section, we investigate the performance of the RPA decoder on second-order RM codes, i.e., codes of the form RM$(m,2)$, $m\geq 2$, which are used for transmission over a BSC$(p)$. Such codes form the ``base case" in our recursive analysis of the performance of RPA decoding, in that they are the smallest codes that the RPA decoder performs both the FHT decoding and aggregation steps on. Our interest is hence in establishing bounds on the probability that the (iterative) RPA decoder converges to the true transmitted codeword. 

Let $\mathbf{Y}\in \{0,1\}^N$ denote the (random) vector received by passing the input codeword $\mathbf{0}\in \text{RM}(m,2)$ through the BSC$(p)$ (recall the argument in Section \ref{sec:main-result} that shows that we can restrict attention to the all-zeros input codeword). We hence have that $\mathbf{Y}\sim \text{Ber}^{\otimes{N}}(p)$. Further, by the definition of the (random) vectors $\mathbf{Y}_{/\mathbb{B}_i}$, for any $i\in [N-1]$, it follows that $\mathbf{Y}_{/\mathbb{B}_i}\sim \text{Ber}^{\otimes (N/2)}(2p(1-p))$. We recall here that the number of projections is $n_{1,m} = N-1$.

\subsection{The Performance of the FHT Decoder}
Recall from Section \ref{sec:rpadef} that the FHT-based decoder for first-order RM codes is simply an efficient implementation of ML decoding. In this section, we hence directly analyze the behaviour of the ML decoder on first-order RM codes, via a slightly different perspective that is closely related to, but different in form from, the FHT-based decoder. We mention that \cite{dumer_burnashev} also presents (tight) upper bounds on the error probability of ML decoding of first-order RM codes; however, we adopt a simpler analysis that is based on standard concentration inequalities.

Now, fix a one-dimensional subspace $\mathbb{B}_i$, for some $i\in [N-1]$, as in Algorithm \ref{alg:rpa}; let $\mathbf{b}_i\in \{0,1\}^N$ denote the non-zero vector in $\mathbb{B}_i$. Given the received vector $\mathbf{Y}$, we have that $\mathbf{Y}_{/\mathbb{B}_i}$ is a noisy version of a codeword in $\text{RM}(m-1,1)$. Consider the ML (or FHT-based) decoder in Step \ref{step:fht-high} of Algorithm \ref{alg:rpa} that acts on $\mathbf{Y}_{/\mathbb{B}_i}$. First, we transform $\mathbf{Y}_{/\mathbb{B}_i}\in \{0,1\}^{N/2}$ to the vector $\mathbf{Y}^{\pm}_{/\mathbb{B}_i}\in \{-1,1\}^{N/2}$, via the standard mapping in Section \ref{sec:notation}.

Further, consider the family of functions $\left(\chi_{\mathbf{s}}: \mathbf{s}\in \{0,1\}^{m-1}\right)$, where $\chi_{\mathbf{s}}(\mathbf{x}):=(-1)^{\mathbf{x}\cdot \mathbf{s}}$, $\mathbf{x}\in \{0,1\}^{m-1}$. We interchangably treat each $\chi_\mathbf{s}$ (and likewise, $\mathbf{Y}_{/\mathbb{B}_i}$) as a function from $\{0,1\}^{m-1}$ to $\{-1,1\}$, and as a length-$N/2$ vector over $\{-1,1\}$. It can easily be checked that there exists a one-one correspondence between codewords of RM$(m-1,1)$ and the collection of vectors $\chi:= \left(\chi_{\mathbf{s}}: \mathbf{s}\in \{0,1\}^{m-1}\right)\cup \left(-\chi_{\mathbf{s}}: \mathbf{s}\in \{0,1\}^{m-1}\right)$. Note that the all-zeros codeword $\mathbf{0}\in \text{RM}(m-1,1)$ corresponds to the vector $+\chi_{\mathbf{0}}$. Thus, the ML decoder \textsf{ML} for the first-order RM$(m-1,1)$ code, given $\mathbf{Y}^\pm_{/\mathbb{B}_i}$, essentially performs:
\begin{equation}
	\label{eq:ml}
	\mathsf{ML}(\mathbf{Y}^\pm_{/\mathbb{B}_i}) = \argmin\limits_{f\in \chi}\  \text{dist}(\mathbf{Y}^\pm_{/\mathbb{B}_i},f),
\end{equation}
where for functions $f, g: \{0,1\}^n\to \{-1,1\}$, we define
\[
\text{dist}(f,g) := \frac{1}{2^n}\cdot \sum_{\mathbf{x}\in \{0,1\}^n} \mathds{1}\{f(\mathbf{x})\neq g(\mathbf{x})\}.
\]
Further, given $f, g: \{0,1\}^n\to \{-1,1\}$, we define their inner product
\[
\langle f,g\rangle:= \frac{1}{2^n}\cdot \sum_{\mathbf{x}\in \{0,1\}^n} f(\mathbf{x})g(\mathbf{x}).
\]
Via simple manipulations (see, e.g., \cite[Prop. 1.9]{donnell}), we have the following lemma:
\begin{lemma}
	For functions $f, g: \{0,1\}^n\to \{-1,1\}$, we have
	\[
\langle f,g\rangle	= 1-2\cdot\text{\normalfont dist}(f,g).
	\]
\end{lemma}
Thus, the task in \eqref{eq:ml} reduces to:
\begin{align*}
	\mathsf{ML}(\mathbf{Y}^\pm_{/\mathbb{B}_i}) &= \argmax\limits_{f\in \chi}\ \langle \mathbf{Y}^\pm_{/\mathbb{B}_i},f\rangle, \end{align*}
which, in turn, identifies $\mathbf{s}\in \{0,1\}^{m-1}$ and $\sigma\in \{-1,1\}$ that maximize
$
	  \langle \mathbf{Y}^\pm_{/\mathbb{B}_i},\sigma\cdot \chi_\mathbf{s}\rangle.
$
We hence write
\begin{equation}
	\label{eq:ml2}
	\mathsf{ML}(\mathbf{Y}^\pm_{/\mathbb{B}_i}) = \argmax\limits_{\mathbf{s}\in \{0,1\}^{m-1},\ \sigma\in \{-1,1\}} \langle \mathbf{Y}^\pm_{/\mathbb{B}_i},\sigma\cdot \chi_\mathbf{s}\rangle.
\end{equation}
As mentioned earlier, the all-zeros codeword corresponds to $\mathbf{s} = \mathbf{0}$ and $\sigma = 1$. We use the expression \eqref{eq:ml2} above as the structure of our ML decoder, in what follows. We mention here that for a function $f:\{0,1\}^n\to \{-1,1\}$, the term $\langle f,\chi_\mathbf{s}\rangle$ is precisely the Fourier coefficient $\widehat{f}(\mathbf{s})$ (see \cite{donnell} for more details on the Fourier analysis of Boolean functions) of $f$ at the vector $\mathbf{s}\in \{0,1\}^{m-1}$.

The following key theorem shows that with very high probability, for sufficiently large $N$, the ML decoder, when acting on $\mathbf{Y}^\pm_{/\mathbb{B}_i}$, produces the correct input codeword $\mathbf{c}\in \text{RM}(m-1,1)$, which is the all-zeros codeword. Let $\eta = \eta(p):= \frac{1}{2}\cdot (1-4p(1-p))$.

\begin{theorem}
	\label{thm:decode1}
	We have that for any $\epsilon<\eta(p)$,
	\[
	\Pr[\mathsf{ML}(\mathbf{Y}^\pm_{/\mathbb{B}_i}) = (\mathbf{0},1)] \geq 1-8N\cdot e^{-\frac{N\epsilon^2}{8}}.
	\]
\end{theorem}
%\begin{proof}
%	Fix any $\epsilon>0$ such that $\epsilon<\eta(p)$.
%\end{proof}
The proof of Theorem \ref{thm:decode1} makes use of the following simple helper lemmas.
\begin{lemma}
	\label{lem:decodezero}
	For all $\epsilon>0$ and any $i\in [N-1]$, we have that 
	\[
	\Pr\left[\lvert\langle \mathbf{Y}^\pm_{/\mathbb{B}_i},\chi_\mathbf{0} \rangle - (1-4p(1-p))\rvert\leq \epsilon\right]\geq 1-2e^{-\frac{N\epsilon^2}{4}}.
	\]
\end{lemma}
\begin{proof}
	Note that
	\begin{align*}
		\langle \mathbf{Y}^\pm_{/\mathbb{B}_i},\chi_\mathbf{0} \rangle &= \frac{1}{{N/2}}\cdot \sum_{\mathbf{x}\in \{0,1\}^{m-1}} \left({Y}^\pm_{/\mathbb{B}_i}\right)_{\mathbf{x}}\\
		&= \frac{2}{N}\cdot \sum_{\mathbf{x}\in \{0,1\}^{m-1}} Z_{\mathbf{x}},
	\end{align*}
where $Z_{\mathbf{x}}\stackrel{\text{i.i.d.}}{\sim} P_Z$, where $P_Z$ is a Rademacher distribution with $P_Z(-1) = 1-P_Z(1) = 2p(1-p)$. Clearly, we have that $\E\left[\langle \mathbf{Y}^\pm_{/\mathbb{B}_i},\chi_\mathbf{0} \rangle\right] = 1-4p(1-p)$. The claim then follows as a  direct consequence of Hoeffding's inequality \cite{hoeffding} (see also \cite[Thm. 2.2.6]{hdp}).
\end{proof}
\begin{lemma}
	\label{lem:decodenotzero}
	For all $\epsilon>0$ and any $i\in [N-1]$, we have that for any $\mathbf{s}\neq \mathbf{0}$,
	\[
	\Pr\left[\lvert\langle \mathbf{Y}^\pm_{/\mathbb{B}_i},\chi_\mathbf{s} \rangle \rvert\leq \epsilon\right]\geq 1-4e^{-\frac{N\epsilon^2}{8}}.
	\]
\end{lemma}
\begin{proof}
	Using the fact that $\mathbf{Y}_{/\mathbb{B}_i}\sim \text{Ber}^{\otimes (N/2)}(2p(1-p))$, we see that
	\begin{align*}
		\langle \mathbf{Y}^\pm_{/\mathbb{B}_i},\chi_\mathbf{s} \rangle &= \frac{2}{N}\cdot \sum_{\mathbf{x}\in \{0,1\}^{m-1}} \left({Y}^\pm_{/\mathbb{B}_i}\right)_{\mathbf{x}}\cdot (-1)^{\mathbf{x}\cdot \mathbf{s}}\\
		&= \frac{2}{N}\cdot \left[\sum_{j=1}^{N/4} Z'_j - \sum_{k=1}^{N/4} Z''_k\right],
	\end{align*}
	where $Z'_j, Z''_k\stackrel{\text{i.i.d.}}{\sim} P_Z$, where $P_Z$ is a Rademacher distribution with $P_Z(-1) = 1-P_Z(1) = 2p(1-p)$. The second equality above holds since the vector $\chi_\mathbf{s}$, for $\mathbf{s}\neq \mathbf{0}$ has exactly $N/4$ coordinates that are $-1$ and $N/4$ coordinates that are $+1$. Let $\alpha_1:= \frac{2}{N}\cdot\sum_{j=1}^{N/4} Z'_j$ and $\alpha_2:= \frac{2}{N}\cdot\sum_{j=1}^{N/4} Z''_j$. Note that $\E[\alpha_1] = \E[\alpha_2] = \frac12\cdot (1-4p(1-p)):= \beta$. Now, we have that
	\begin{align*}
		\Pr\left[\lvert\langle \mathbf{Y}^\pm_{/\mathbb{B}_i},\chi_\mathbf{s} \rangle \rvert\leq \epsilon\right]
		&= \Pr[|\alpha_1-\alpha_2|\leq \epsilon]\\
		&\geq \Pr\left[\alpha_1\in [\beta-\epsilon/2,\beta+\epsilon/2]\text{ and } \alpha_2\in [\beta-\epsilon/2,\beta+\epsilon/2]\right]\\
		&\geq \left(1-2e^{-\frac{N\epsilon^2}{8}}\right)^2 \geq 1-4e^{-\frac{N\epsilon^2}{8}},
	\end{align*}
	 where the second inequality holds via an application of Hoeffding's inequality \cite{hoeffding} (see also \cite[Thm. 2.2.6]{hdp}).
\end{proof}
The proof of Theorem \ref{thm:decode1} then follows easily by combining Lemmas \ref{lem:decodezero} and \ref{lem:decodenotzero} above.
\begin{proof}[Proof of Theorem \ref{thm:decode1}]
	Fix any $0<\epsilon< \eta(p)$. First, observe that by a union bound argument, the intersection of the events 
	\[
	\lvert\langle \mathbf{Y}^\pm_{/\mathbb{B}_i},\chi_\mathbf{0} \rangle - (1-4p(1-p))\rvert\leq \epsilon
	\]
	and $\lvert\langle \mathbf{Y}^\pm_{/\mathbb{B}_i},\chi_\mathbf{s} \rangle \rvert\leq \epsilon$, for all $\mathbf{s}\neq \mathbf{0}$, occurs with probability at least $1-8N\cdot e^{-\frac{N\epsilon^2}{8}}.$ In other words, we have that with probability at least $1-8N\cdot e^{-\frac{N\epsilon^2}{8}}$, the following events occur: $\langle \mathbf{Y}^\pm_{/\mathbb{B}_i},\chi_\mathbf{0} \rangle \in [(1-4p(1-p))-\epsilon,(1-4p(1-p))+\epsilon]$ and $\langle \mathbf{Y}^\pm_{/\mathbb{B}_i},-\chi_\mathbf{0} \rangle \in [-(1-4p(1-p))-\epsilon,-(1-4p(1-p))+\epsilon]$; furthermore, for all $\mathbf{s}\neq \mathbf{0}$, we have $\langle \mathbf{Y}^\pm_{/\mathbb{B}_i},\chi_\mathbf{s} \rangle \in [-\epsilon,\epsilon]$ and $\langle \mathbf{Y}^\pm_{/\mathbb{B}_i},-\chi_\mathbf{s} \rangle \in [-\epsilon,\epsilon]$. From our choice of $\epsilon< \eta(p)$, it can thus be verified that with probability at least $1-8N\cdot e^{-\frac{N\epsilon^2}{8}}$, we have
	\[
	\langle \mathbf{Y}^\pm_{/\mathbb{B}_i},\chi_\mathbf{0} \rangle = \max_{f\in \chi} \langle \mathbf{Y}^\pm_{/\mathbb{B}_i},f\rangle.
	\]
	thereby concluding the proof of the theorem.
\end{proof}
With Theorem \ref{thm:decode1} in place, a simple application of the union bound then results in the following result:
\begin{proposition}
	\label{prop:decode2}
	We have that for any $\epsilon<\eta(p)$,
	\begin{align*}
	&\Pr\left[\mathsf{ML}(\mathbf{Y}^\pm_{/\mathbb{B}_i}) = (\mathbf{0},1),\text{\normalfont for all }i\in [N-1]\right] \geq 1-8N(N-1)\cdot e^{-\frac{N\epsilon^2}{8}}.
	\end{align*}
In other words, we have that for any $\epsilon<\eta(p)$,
\begin{align*}
	&\Pr\left[\mathbf{\widehat{Y}}_{/\mathbb{B}_i} = \mathbf{0},\text{\normalfont for all }i\in [N-1]\right]\geq 1-8N(N-1)\cdot e^{-\frac{N\epsilon^2}{8}}.
\end{align*}
\end{proposition}

We shall use this fact in the next subsection to argue that with very high probability, for sufficiently large $N$, the aggregation step produces the correct, all-zeros input codeword in just a single iteration of the RPA decoder.

\subsection{The Aggregation Step}
\label{sec:agg}
Let $\mathcal{G}$ denote the event $\left\{\mathbf{\widehat{Y}}_{/\mathbb{B}_i} = \mathbf{0},\text{\normalfont for all }i\in [N-1]\right\}$. Proposition \ref{prop:decode2} shows that $\Pr[\mathcal{G}]\geq 1-8N(N-1)\cdot e^{-\frac{N\epsilon^2}{8}}.$ In this subsection, we analyze the performance of the aggregation step, conditioned on the event $\mathcal{G}$. Observe that, conditioned on $\mathcal{G}$, the aggregation step computes, for each $\mathbf{x}\in \{0,1\}^m$, the function $\phi(\mathbf{x}) =: \phi^{(N)}(\mathbf{x}) = \sum_{i=1}^{N-1} \mathds{1}\{Y_{/\mathbb{B}_i}([\mathbf{x}+\mathbb{B}_i])\neq \widehat{Y}_{/\mathbb{B}_i}([\mathbf{x}+\mathbb{B}_i])\}$. Upon conditioning on $\mathcal{G}$, it can be checked that $\phi(\mathbf{x})$ obeys
\[
\phi(\mathbf{x}) =  \sum_{i=1}^{N-1} \left(Y_\mathbf{x}\oplus Y_{\mathbf{x}+\mathbf{b}_i}\right),
\]
where the summation over $i\in [N-1]$ above is over the reals. The aggregation step then sets \textsf{Flip}$(\mathbf{x}) = \textsf{Flip}^{(N)}(\mathbf{x}) = 1$, if $\phi(\mathbf{x}) > \frac{N-1}{2}$, and \textsf{Flip}$(\mathbf{x}) = 0$, otherwise. Finally, the vector $\overline{\mathbf{Y}} = \mathbf{Y}\oplus \textsf{Flip}$ is computed, which then serves as the new ``received vector" in the next iteration.

In what follows, we argue that with high probability, for large enough $N$, the vector $\textsf{Flip}$ equals the vector $\mathbf{Y}$,  thereby leading to $\overline{\mathbf{Y}} = \mathbf{0}$. It can then be checked that once $\overline{\textbf{Y}}$ equals the all-zeros vector at the end of the first iteration of the RPA decoder, further iterations only produce the all-zeros vector, implying that the RPA decoder has converged to the correct input codeword after just one iteration.

Now, observe that
\begin{align*}
	\phi(\mathbf{x}) = \begin{cases}
		\sum_{\mathbf{z}\neq \mathbf{x}} Y_\mathbf{z},\ \text{if $Y_\mathbf{x} = 0$},\\
		N-1-\sum_{\mathbf{z}\neq \mathbf{x}} Y_\mathbf{z},\ \text{if $Y_\mathbf{x} = 1$}.
	\end{cases}
\end{align*}
where each $Y_\mathbf{z}\stackrel{\text{i.i.d.}}{\sim} \text{Ber}(p)$. We define $\overline{\phi}(\mathbf{x}) = \overline{\phi}^{(N)}(\mathbf{x}) = \frac{\phi(\mathbf{x})}{N-1}$ and $\overline{\phi}_\infty(\mathbf{x}):= p(1-Y_\mathbf{x})+(1-p)Y_\mathbf{x}$. It is not too hard to believe that for large enough $N$, conditioned on the event $\mathcal{G}$, the random variable $\textsf{Flip}^{(N)}(\mathbf{x}) = \mathds{1}\{\overline{\phi}(\mathbf{x})>\frac12\}$ is close to  $\mathds{1}\{\overline{\phi}_\infty(\mathbf{x})>\frac12\}$. The following proposition thus holds:
\begin{proposition}
	\label{prop:flip}
	For all $\epsilon<\eta(p)$, we have that for any $\mathbf{x}\in \{0,1\}^m$, 
	\begin{align*}
	&\Pr\left[\mathsf{Flip}^{(N)}(\mathbf{x}) = \mathds{1}\left\{\overline{\phi}_\infty(\mathbf{x})>\frac12\right\}\bigg\vert\ \mathcal{G} \right] \geq  1-16N(N-1)\cdot e^{-\frac{N\epsilon^2}{8}}.
	\end{align*}
\end{proposition}
Before we prove the above proposition, we state and prove a useful lemma:
\begin{lemma}
	\label{lem:flip}
	For any $\epsilon>0$, we have that for any $\mathbf{x}\in \{0,1\}^m$,
	\[
	\Pr\left[\left \lvert \overline{\phi}^{(N)}(\mathbf{x}) - \overline{\phi}_\infty(\mathbf{x}) \right\rvert \leq  \epsilon\right]\geq 1-2\cdot e^{-2(N-1)\epsilon^2}.
	\]
\end{lemma}
\begin{proof}
	First, observe that $\E\left[ Y_\mathbf{z}\right] = p$, for any $\mathbf{z}\in \{0,1\}^m$. Further, we can write, for any $\mathbf{x}\in \{0,1\}^m$,
	\begin{align}
	\overline{\phi}^{(N)}(\mathbf{x}) &= \left(\frac{1}{N-1} \sum_{\mathbf{z}\neq \mathbf{x}}Y_\mathbf{z}\right)\cdot (1-Y_\mathbf{x})+Y_\mathbf{x}\cdot \left(1-\frac{1}{N-1} \sum_{\mathbf{z}\neq \mathbf{x}}Y_\mathbf{z}\right).
	\label{eq:temphoeff}
	\end{align}
	Now, conditioned on any fixed value $y_\mathbf{x}\in \{0,1\}$, we have via Hoeffding's inequality \cite{hoeffding} (see also \cite[Thm. 2.2.6]{hdp}), that with probability at least $1-2\cdot e^{-2(N-1)\epsilon^2}$, 
	\[
	\left \lvert \frac{1}{N-1}\cdot \sum_{\mathbf{z}\neq \mathbf{x}}Y_\mathbf{z} - p\right \rvert \leq \epsilon,
	\]
	or equivalently, that
	\[
	\left \lvert \left(1- \frac{1}{N-1}\cdot \sum_{\mathbf{z}\neq \mathbf{x}}Y_\mathbf{z} \right)- (1-p)\right \rvert \leq \epsilon.
	\]
	We mention that the above result uses the fact that $Y_{\mathbf{x}}$ is independent of $\{Y_\mathbf{z}\}_{\mathbf{z}\neq \mathbf{x}}$. Substituting this back in \eqref{eq:temphoeff} and using the law of total probability, we obtain that with probability at least $1-2\cdot e^{-2(N-1)\epsilon^2}$, we must have
	$
	\left \lvert \overline{\phi}^{(N)}(\mathbf{x}) - \overline{\phi}_\infty(\mathbf{x}) \right\rvert \leq  \epsilon,
	$
	thereby proving the lemma.
%	Hence,
%	\begin{align*}
%	\overline{\phi}^{(N)}(\mathbf{x}) - \overline{\phi}_\infty(\mathbf{x}) &= \frac{1}{N-1}\cdot \sum_{\mathbf{z}\neq \mathbf{x}} Y_\mathbf{z} - p
%	\end{align*}
%is such that $\left \lvert \overline{\phi}^{(N)}(\mathbf{x}) - \overline{\phi}_\infty(\mathbf{x}) \right\rvert\leq  \epsilon$ with probability at least $1-2\cdot e^{\frac{(N-1)\epsilon^2}{2}}$, 
\end{proof}
While the lemma above shows that with high probability, the random variables inside the indicator functions in $\mathsf{Flip}^{(N)}(\mathbf{x})$ and $\mathds{1}\left\{\overline{\phi}_\infty(\mathbf{x})>\frac12\right\}$ are close, for large enough $N$, Proposition \ref{prop:flip} claims that these indicator functions are themselves close, with high probability, when conditioned on the event $\mathcal{G}$. We now prove Proposition \ref{prop:flip}. 

Before we proceed, we introduce one additional piece of notation: let $\overline{\mathcal{G}}$ denote the event $\left\{\left \lvert \overline{\phi}^{(N)}(\mathbf{x}) - \overline{\phi}_\infty(\mathbf{x}) \right\rvert \leq  \epsilon\right\}$.

\begin{proof}[Proof of Proposition \ref{prop:flip}]
	We shall first prove that unconditionally, for any $\mathbf{x}\in \{0,1\}^m$, the random variable $\mathds{1}\left\{\overline{\phi}(\mathbf{x})>\frac12\right\}$ is close to the random variable $\mathds{1}\left\{\overline{\phi}_\infty(\mathbf{x})>\frac12\right\}$, with high probability, if $N$ is sufficiently large. This will then allow us to deduce the result of the proposition, by noting that conditioned on the event $\mathcal{G}$, we have that $\textsc{Flip}^{(N)}(\mathbf{x})$ equals $\mathds{1}\left\{\overline{\phi}(\mathbf{x})>\frac12\right\}$. 
	
	To this end, observe that the following sequence of inequalities holds:
	\begin{align*}
		%&\Pr\left[\left\lvert\mathsf{Flip}^{(N)}(\mathbf{x}) - \mathds{1}\left\{\overline{\phi}_\infty(\mathbf{x})>\frac12\right\}\right \rvert\leq \epsilon \right]\\
		\Pr\left[\left\lvert\mathds{1}\left\{\overline{\phi}(\mathbf{x})>\frac12\right\} - \mathds{1}\left\{\overline{\phi}_\infty(\mathbf{x})>\frac12\right\}\right \rvert\leq \epsilon \right]
		&\stackrel{(a)}{\geq}\Pr\left[\left\lvert\mathds{1}\left\{\overline{\phi}(\mathbf{x})>\frac12\right\} - \mathds{1}\left\{\overline{\phi}_\infty(\mathbf{x})>\frac12\right\}\right \rvert\leq \epsilon\ \bigg\vert \ \overline{\mathcal{G}} \right]\cdot \Pr[\overline{\mathcal{G}}]\\
		&\stackrel{(b)}{=}\Pr\left[\left\lvert\mathds{1}\left\{\overline{\phi}(\mathbf{x})>\frac12\right\} - \mathds{1}\left\{\overline{\phi}_\infty(\mathbf{x})>\frac12\right\}\right \rvert\leq \epsilon\ \bigg\vert\  \overline{\mathcal{G}} \right]\cdot(1-\delta),
	\end{align*}
where $\delta :=  2e^{-2(N-1)\epsilon^2}.$ Here, inequality (a) holds via the law of total probability, and (b) holds via Lemma \ref{lem:flip}. We now claim that conditioned on $\overline{\mathcal{G}}$, the event 
\[
\mathcal{A}':= \left\{\left\lvert\mathds{1}\left\{\overline{\phi}_\infty(\mathbf{x})-\epsilon>\frac12\right\} - \mathds{1}\left\{\overline{\phi}_\infty(\mathbf{x})>\frac12\right\}\right \rvert\leq \epsilon\right\}
\]
implies the event
\[
\mathcal{A}:= \left\{\left\lvert\mathds{1}\left\{\overline{\phi}(\mathbf{x})>\frac12\right\} - \mathds{1}\left\{\overline{\phi}_\infty(\mathbf{x})>\frac12\right\}\right \rvert\leq \epsilon\right\}.
\]
To this end, it suffices to show that the expression $g_1(\mathbf{x}):= \left\lvert\mathds{1}\left\{\overline{\phi}(\mathbf{x})>\frac12\right\} - \mathds{1}\left\{\overline{\phi}_\infty(\mathbf{x})>\frac12\right\}\right \rvert$ is at most the expression $g_2(\mathbf{x}):= \left\lvert\mathds{1}\left\{\overline{\phi}_\infty(\mathbf{x})-\epsilon>\frac12\right\} - \mathds{1}\left\{\overline{\phi}_\infty(\mathbf{x})>\frac12\right\}\right \rvert$, if $\overline{\mathcal{G}}$ holds. Indeed, observe that since $\epsilon<1/2$, we have from the event $\overline{\mathcal{G}}$ that
\[
\overline{\phi}_\infty(\mathbf{x})-\frac12\geq \overline{\phi}(\mathbf{x})-\frac12\geq \overline{\phi}_\infty(\mathbf{x})-\epsilon-\frac12,
\]
thereby leading to
\[
\mathds{1}\left\{\overline{\phi}_\infty(\mathbf{x})>\frac12\right\}\leq \mathds{1}\left\{\overline{\phi}(\mathbf{x})>\frac12\right\}\leq \mathds{1}\left\{\overline{\phi}_\infty(\mathbf{x})-\epsilon>\frac12\right\}.
\]
Hence, we have that 
\begin{align*}g_1(\mathbf{x}) &= \mathds{1}\left\{\overline{\phi}(\mathbf{x})>\frac12\right\} - \mathds{1}\left\{\overline{\phi}_\infty(\mathbf{x})>\frac12\right\}\\
	&\leq \mathds{1}\left\{\overline{\phi}_\infty(\mathbf{x})-\epsilon>\frac12\right\} - \mathds{1}\left\{\overline{\phi}_\infty(\mathbf{x})>\frac12\right\} = g_2(\mathbf{x}),\end{align*}
thereby proving the claim.

Hence, carrying on from (b) above, we get that since the event $\mathcal{A}'$ implies the event $\mathcal{A}$, conditioned on $\overline{\mathcal{G}}$, 
\begin{align*}
	&\Pr\left[\left\lvert\mathds{1}\left\{\overline{\phi}(\mathbf{x})>\frac12\right\} - \mathds{1}\left\{\overline{\phi}_\infty(\mathbf{x})>\frac12\right\}\right \rvert\leq \epsilon \right]\\
	&\geq \Pr\left[\left\lvert\mathds{1}\left\{\overline{\phi}_\infty(\mathbf{x})-\epsilon>\frac12\right\} - \mathds{1}\left\{\overline{\phi}_\infty(\mathbf{x})>\frac12\right\}\right \rvert\leq \epsilon \right]\cdot (1-\delta)\\
	&=\Pr\left[\left\lvert \mathds{1}\left\{Y_\mathbf{x}>\frac{\frac12+\epsilon-p}{1-2p}\right\} - \mathds{1}\left\{Y_\mathbf{x}>\frac{\frac12-p}{1-2p}\right\} \right \rvert\leq \epsilon \right]\cdot (1-\delta)\\
	&= \Pr\left[ \mathds{1}\left\{Y_\mathbf{x}>\frac{\frac12+\epsilon-p}{1-2p}\right\} = \mathds{1}\left\{Y_\mathbf{x}>\frac12\right\} \right]\cdot (1-\delta)\\
	&= \left(\Pr\left[Y_\mathbf{x}<\frac12\right]+\Pr\left[Y_\mathbf{x}>\frac{\frac12+\epsilon-p}{1-2p}\right]\right)\cdot(1-\delta) = 1-\delta,
\end{align*}
since $0<\epsilon<\eta(p)\leq \frac12-p$, for all $p<1/2$.

Now, let $\delta':= 8N(N-1)\cdot e^{-\frac{N\epsilon^2}{8}}$. The following inequalities then hold:
\begin{align*}
	&\Pr\left[\left\lvert\mathsf{Flip}^{(N)}(\mathbf{x}) - \mathds{1}\left\{\overline{\phi}_\infty(\mathbf{x})>\frac12\right\}\right \rvert\leq \epsilon\ \bigg\vert\ \mathcal{G} \right]\\
	&=\Pr\left[\left\lvert\mathds{1}\left\{\overline{\phi}(\mathbf{x})>\frac12\right\} - \mathds{1}\left\{\overline{\phi}_\infty(\mathbf{x})>\frac12\right\}\right \rvert\leq \epsilon\ \bigg\vert\ \mathcal{G}  \right]\\
	&\geq  {\Pr\left[\left\lvert\mathds{1}\left\{\overline{\phi}(\mathbf{x})>\frac12\right\} - \mathds{1}\left\{\overline{\phi}_\infty(\mathbf{x})>\frac12\right\}\right \rvert\leq \epsilon \right] - \Pr[\mathcal{G}^c]}\\
	&\geq 1-\delta-\delta'\geq 1-2\delta',
\end{align*}
since $\delta'\geq \delta$, for all $N\geq 2$. Here, we define $\mathcal{G}^c:= \left\{\widehat{Y}_{/\mathbb{B}_i} \neq \mathbf{0},\text{\normalfont for some }i\in [N-1]\right\}$, which leads to the first inequality above by simple manipulations of the law of total probability. The second inequality above then makes use of Proposition \ref{prop:decode2}.

The proof of the proposition then follows by observing that  $\mathsf{Flip}^{(N)}(\mathbf{x})$ and $\mathds{1}\left\{\overline{\phi}_\infty(\mathbf{x})>\frac12\right\}$ are both Bernoulli random variables, and hence their closeness must imply their equality, since $\epsilon<\eta(p)<1$.
\end{proof}
%We have now proved that for large enough $N$, conditioned on the event $\mathcal{G}$, the random variables $\mathsf{Flip}^{(N)}(\mathbf{x})$ and $\mathds{1}\left\{\overline{\phi}_\infty(\mathbf{x})>\frac12\right\}$ are close with very high probability. However, since both these random variables are Bernoulli, it must be that their closeness implies their equality. The following corollary makes this intuition formal:
%\begin{corollary}
%	\label{cor:flip}
%	For all $\epsilon<\min\{\eta(p),\beta(p)\}$, we have
%	\[
%	\Pr\left[\mathsf{Flip}^{(N)}(\mathbf{x}) = \mathds{1}\left\{\overline{\phi}_\infty(\mathbf{x})>\frac12\right\}\bigg\vert \mathcal{G} \right]\geq 1-2\cdot e^{\frac{(N-1)\epsilon^2}{2}}.
%	\]
%\end{corollary}
%\begin{proof}
%	The proof follows by observing that in Proposition \ref{prop:flip}, we have $\epsilon<\eta(p)\leq 1$.
%\end{proof}
Thus, after some manipulation, we have obtained, for any $\mathbf{x}\in \{0,1\}^m$, conditioned on the event $\mathcal{G}$, a simple expression for the structure of the random variable $\mathsf{Flip}^{(N)}(\mathbf{x})$ that holds for sufficiently large $N$, with high probability. All that remains to be shown is that the random vector $\mathsf{Flip}^{(N)}$, with high probability, equals the received vector $\mathbf{Y}$. The following theorem shows that this fact is indeed true:
\begin{theorem}
	\label{thm:agg}
	We have that for  all $\epsilon<\eta(p)$,
	\[
	\Pr\left[\mathsf{Flip}^{(N)} = \mathbf{Y}\right]\geq 1-32N^3\cdot e^{-\frac{N\epsilon^2}{8}}.
	\]
\end{theorem}
\begin{proof}
	First, observe that by a union bound argument, we have from Proposition \ref{prop:flip} that
	\begin{align}
		\Pr\left[\mathsf{Flip}^{(N)}(\mathbf{x}) = \mathds{1}\left\{\overline{\phi}_\infty(\mathbf{x})>\frac12\right\},\text{ for all $\mathbf{x}$}\ \bigg\vert\ \mathcal{G} \right] 
		&\geq 1-16N^2(N-1)\cdot e^{-\frac{N\epsilon^2}{8}}. \label{eq:temp1}
	\end{align}
Next, we note that the random variable $\mathds{1}\left\{\overline{\phi}_\infty(\mathbf{x})>\frac12\right\}$ equals $\mathds{1}\left\{Y_\mathbf{x}>\frac12\right\}$, which, in turn, equals $Y_\mathbf{x}$ itself. Hence, \eqref{eq:temp1} above simply says that $$\Pr\left[\mathsf{Flip}^{(N)} = \mathbf{Y}\  \bigg\vert\ \mathcal{G} \right] \geq 1-16N^2(N-1)\cdot e^{-\frac{N\epsilon^2}{8}}.$$
Finally, note that the unconditional probability
\begin{align}
	\Pr\left[\mathsf{Flip}^{(N)} = \mathbf{Y}\right]&\geq \Pr\left[\mathsf{Flip}^{(N)} = \mathbf{Y}\  \bigg\vert\ \mathcal{G} \right]\cdot \Pr[\mathcal{G}] \label{eq:temp1higher}\\
	&\geq \left(1-16N^2(N-1)\cdot e^{-\frac{N\epsilon^2}{8}}\right)\left(1-8N(N-1)\cdot e^{-\frac{N\epsilon^2}{8}}\right) \notag\\
	&\geq 1-16N^2(N-1)\cdot e^{-\frac{N\epsilon^2}{8}}-8N(N-1)\cdot e^{-\frac{N\epsilon^2}{8}} \notag\\
	&\geq 1-32N^3\cdot e^{-\frac{N\epsilon^2}{8}}. \label{eq:temp2higher}
\end{align}
Here, the second inequality above holds via Proposition \ref{prop:decode2}.
\end{proof}
Thus, we finally obtain the following corollary, which was the statement we were aiming to prove:

\begin{figure*}
	\centering
	\includegraphics[width = 0.75\textwidth]{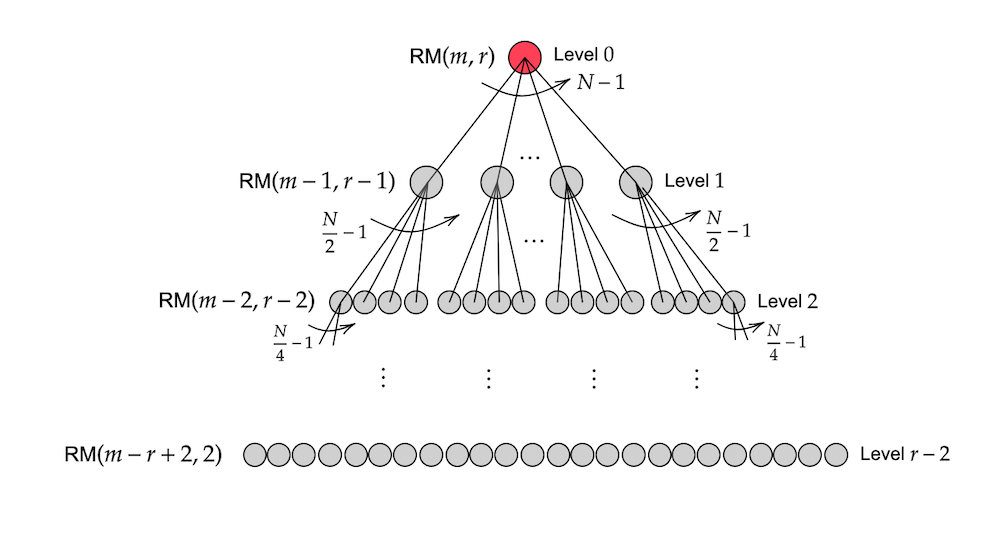}
	\caption{Figure representing a projection-aggregation tree}
	\label{fig:projtree}
\end{figure*}
\begin{corollary}
	\label{cor:agg}
	After a single iteration of the RPA decoder for second-order RM codes, we have for all $\epsilon<\eta(p)$ that
	\[
	\Pr[\overline{\mathbf{Y}} = \mathbf{0}]\geq 1-32N^3\cdot e^{-\frac{N\epsilon^2}{8}}.
	\]
\end{corollary}
Thus, we see that with overwhelming probability, given any codeword from a second-order RM code that is transmitted across a BSC$(p)$, for sufficiently large $N$, a \emph{single} iteration of the RPA decoder, which includes a run of an ML decoder for a first-order RM code and an aggregation step, produces the correct codeword as the decoded estimate.

For second-order RM codes, one can in addition obtain a bound on the error probability after \emph{two} iterations of the RPA decoder, thereby yielding an improvement over the bound in Corollary \ref{cor:agg}. The following lemma makes this precise:
\begin{lemma}
		\label{lem:agg2}
	After two iterations of the RPA decoder for second-order RM codes, we have for all $\epsilon<\eta(p)$ that
	\[
	\Pr[\overline{\mathbf{Y}} = \mathbf{0}]\geq 1-256N^2\cdot e^{-\frac{N\epsilon^2}{8}}.
	\]
\end{lemma}
We thus incur a saving of a multiplicative factor that grows linearly in $N$, as compared to the bound in Corollary \ref{cor:agg}. We now prove Lemma \ref{lem:agg2}.
\begin{proof}[Proof of Lemma \ref{lem:agg2}]
From Proposition \ref{prop:flip}, for any $\mathbf{x}\in \{0,1\}^m$, we have that  $\Pr\left[\mathsf{Flip}^{(N)}(\mathbf{x}) = Y_\mathbf{x} \ \bigg\vert\ \mathcal{G} \right]\geq 1-16N(N-1)\cdot e^{-\frac{N\epsilon^2}{8}}.$ Via Markov's inequality, this then allows us to obtain an upper bound on the probability that, conditioned on the event $\mathcal{G}$, the vector $\textsf{Flip}^{(N)}$ differs from $\mathbf{Y}$ in at least $2^{m-3}$ positions, where $2^{m-3} = N/8= d/2$. Indeed, we have that
\begin{align*}
	\Pr\left[w\left(\textsf{Flip}^{(N)}\oplus \mathbf{Y}\right) \geq N/8 \mid \mathcal{G}\right]&\leq \frac{16N^2(N-1)\cdot e^{-\frac{N\epsilon^2}{8}}}{N/8}\\
	&= 128N(N-1)\cdot e^{-\frac{N\epsilon^2}{8}}.
\end{align*}
By calculations similar to those in \eqref{eq:temp1higher}--\eqref{eq:temp2higher}, we obtain $\Pr\left[w\left(\textsf{Flip}^{(N)}\oplus \mathbf{Y}\right) \geq N/8\right]\geq 1-256N^2e^{-\frac{N\epsilon^2}{8}}$. Now, note that if $\textsf{Flip}^{(N)}$ and $\mathbf{Y}$ differ in at most $N/8-1$ positions, we have that at the end of one RPA iteration, $w\left(\overline{\mathbf{Y}}\right)\leq d_\text{min}(\text{RM}(m,2))/2 -1$. This hence implies, via arguments as in \cite[Lemma 1]{coseterror} that the RPA decoder, in the next iteration, can correct all errors present in $\overline{\mathbf{Y}}$, resulting in convergence to the correct all-zeros codeword, at the end of two iterations.
\end{proof}
While Lemma \ref{lem:agg2} gives rise to a tighter estimate of the error probability of decoding second-order RM codes under RPA decoding, in this work, we persist with the calculations made via Theorem \ref{thm:agg}, since the above reduction in error probabilities leaves the exponent unchanged. Further, as we shall see in the next section, Theorem \ref{thm:agg} can be used directly as an estimate of the error probabilities in the aggregation phase, for RM codes of order $r$ larger than $2$ as well.
\section{Extension to Higher-Order RM Codes}
\label{sec:higher}
Via our analysis in the previous section for second-order RM codes, we obtained estimates (lower bounds) on the probability of correct decoding of the ML decoder for first-order RM codes, and the probability with which a single aggregation step produces the correct estimate $\mathbf{0}$ of a codeword, given that all of the projections of the received vector are decoded to $\mathbf{0}$. These results, as we argue in this section, allow us to obtain estimates of the probability of correct decoding under RPA decoding of general Reed-Muller codes RM$(m,r)$ too, thereby resulting in a proof of Theorem \ref{thm:rpaerror}.

First, we introduce some notation. Fix a Reed-Muller code RM$(m,r)$, for some $m\geq 1$ and $2<r\leq m$. Now, for any $t\in [r]$, let $\mathbb{S}^{(1)},\ldots,\mathbb{S}^{(t)}$ be a sequence of one-dimensional subspaces, where $\mathbb{S}_j\subseteq \{0,1\}^{m-j+1}$, for each $j\in [t]$. Given the (random) received vector $\mathbf{Y}$ obtained by passing the input codeword $\mathbf{0}$ through the BSC$(p)$, we recall from \eqref{eq:project} that 
\begin{equation*}
	\mathbf{Y}_{/\mathbb{S}^{(1)}}:=\left({Y}_{/\mathbb{S}^{(1)}}(T):\ T\in \{0,1\}^m/\mathbb{S}^{(1)}\right).
\end{equation*}
The following quantities are then defined recursively: for every $1\leq j\leq t-1$, we index the coordinates of $\mathbf{Y}_{/\left(\mathbb{S}^{(1)},\mathbb{S}^{(2)},\ldots,\mathbb{S}^{(j)}\right)}$ by vectors $\mathbf{z}\in \{0,1\}^{m-j}$ in lexicographic order, and set
\begin{align}
	\label{eq:projecthigher}
	\mathbf{Y}_{/\left(\mathbb{S}^{(1)},\ldots,\mathbb{S}^{(j+1)}\right)}:=\bigg(&\left({Y}_{/\left(\mathbb{S}^{(1)},\mathbb{S}^{(2)},\ldots,\mathbb{S}^{(j)}\right)}\right)_{/\mathbb{S}^{(j+1)}}(T):  T\in \{0,1\}^{m-j}/\mathbb{S}^{(j+1)}\bigg).
\end{align}
For any (random) codeword $\mathbf{C}\in \text{RM}(m,r)$, one can define projections $\mathbf{C}_{/\left(\mathbb{S}^{(1)},\mathbb{S}^{(2)},\ldots,\mathbb{S}^{(j)}\right)}$, for $1\leq j\leq m,$ analogously; via arguments made in \cite[Lemma 1]{rpa}, it can be shown that $\mathbf{C}_{/\left(\mathbb{S}^{(1)},\mathbb{S}^{(2)},\ldots,\mathbb{S}^{(j)}\right)}\in \text{RM}(m-j,r-j)$, if $\mathbf{C}\in \text{RM}(m,r)$.

The definitions above, along with the recursive structure of the RPA decoder (see Algorithm \ref{alg:rpa}) motivate the definition of a ``projection-aggregation tree"---a simple, visual device that helps us recursively compute the probabilities of error in decoding for higher-order RM codes. Formally, a projection-aggregation tree is a rooted tree, with the root (at level $0$) being the code RM$(m,r)$. For every level $0\leq i\leq r-3$, each node at level $i$ has $\frac{N}{2^i}-1$ children at level $i+1$, each of which is an RM$(m-i-1,r-i-1)$ code; here, we recall that $N=2^m$. Clearly, the leaves of the projection-aggregation tree are RM$(m-r+2,2)$ codes. Figure \ref{fig:projtree} shows a pictorial depiction of a projection-aggregation tree. Each path in the tree from the root to a leaf node corresponds to a fixed sequence $\mathbb{S}_1,\ldots,\mathbb{S}_{r-2}$ of one-dimensional subspaces as above. ``Moving down" from a parent node to a child node in this tree signifies a projection step and ``moving up" from a child node to its parent signifies an aggregation step.

Now, for any $j\in [r-2]$, recall that $p^{(j)}= \frac12\cdot (1-(1-2p)^{2^j})$ and let $N^{(j)}:=N/2^j$. Since $\mathbf{Y}\sim \text{Ber}^{\otimes N}(p)$, the following simple lemma holds:
\begin{lemma}
	\label{lem:ptree}
	For any choice of one-dimensional subspaces $\mathbb{S}^{(1)},\ldots,\mathbb{S}^{(j)}$, for $j\in [r-2]$, where $\mathbb{S}_k\subseteq \{0,1\}^{m-k+1}$, for $k\in [j]$, we have
	\[
	\mathbf{Y}_{/\left(\mathbb{S}^{(1)},\ldots,\mathbb{S}^{(j)}\right)}\sim \text{\normalfont Ber}^{\otimes (N^{(j)})}(p^{(j)}).
	\]
\end{lemma}
\begin{proof}
	First, observe that the vector $\mathbf{Y}_{/\left(\mathbb{S}^{(1)},\ldots,\mathbb{S}^{(j)}\right)}$ is of length $N/2^j = N^{(j)}$, by construction. Further, by the disjointness of distinct cosets, we obtain from \eqref{eq:projecthigher} that each coordinate of $\mathbf{Y}_{/\left(\mathbb{S}^{(1)},\ldots,\mathbb{S}^{(j)}\right)}$ is the sum, modulo $2$, of $2^j$ \emph{distinct} coordinates of $\mathbf{Y}$. This implies that any coordinate of $\mathbf{Y}_{/\left(\mathbb{S}^{(1)},\ldots,\mathbb{S}^{(j)}\right)}$ takes the value $1$ with probability equal to the probability that the parity, or the sum modulo $2$ of $2^j$ Ber$(p)$ random variables, equals $1$, which in turn equals $p^{(j)}$. Moreover, distinct coordinates of $\mathbf{Y}_{/\left(\mathbb{S}^{(1)},\ldots,\mathbb{S}^{(j)}\right)}$ correspond to modulo $2$ sums of \emph{disjoint} collections of coordinates in $\mathbf{Y}$, thereby yielding the statement of the lemma.
\end{proof}
Now, consider any leaf node in the projection-aggregation tree, corresponding to a second-order RM code RM$(m-r+2,2)$, obtained via some fixed path of subspaces $\overline{\mathbb{S}}_1,\ldots,\overline{\mathbb{S}}_{r-2}$. Let $\textsc{Agg}^{(r-2)}\left(\overline{\mathbb{S}}^{(1)},\ldots,\overline{\mathbb{S}}^{(r-2)}\right)$ denote the vector $\overline{\mathbf{Y}}$ obtained at the node of interest, post aggregation using the decoded estimates of the projections $$\left(\mathbf{Y}_{/\left(\overline{\mathbb{S}}^{(1)},\ldots,\overline{\mathbb{S}}^{(r-2)},\mathbb{S}^{(r-1)}\right)}:\ \mathbb{S}^{(r-1)}\in \{0,1\}^{m-r+2}\right).$$
Recall that the projections above were decoded using the ML decoder for first-order RM codes. Using the fact (from Lemma \ref{lem:ptree}) that $\mathbf{Y}_{/\left(\overline{\mathbb{S}}^{(1)},\ldots,\overline{\mathbb{S}}^{(r-2)}\right)}\sim \text{Ber}^{\otimes(N^{(r-2)})}(p^{(r-2)})$, we obtain via our analysis of the error probabilities for second-order RM codes (in particular, Corollary \ref{cor:agg}) that for any $\epsilon<\eta(p^{(r-2)})$,
\begin{align*}
	\Pr&\left[\textsc{Agg}^{(r-2)}\left(\overline{\mathbb{S}}^{(1)},\ldots,\overline{\mathbb{S}}^{(r-2)}\right)= \mathbf{0} \right]\geq 1-32\left(N^{(r-2)}\right)^3\cdot e^{-\frac{N^{(r-2)}\epsilon^2}{8}}.
\end{align*}
Now, consider the (unique) parent of the child node of interest, which corresponds to the path $\overline{\mathbb{S}}^{(1)},\ldots,\overline{\mathbb{S}}^{(r-3)}$. Observe that the vectors $\textsc{Agg}^{(r-2)}\left(\overline{\mathbb{S}}^{(1)},\ldots,{\mathbb{S}}^{(r-2)}\right)$ are precisely the $\widehat{\mathbf{Y}}_{/{\mathbb{S}}^{(r-2)}}$ vectors that are passed ``upwards" in the projection-aggregation tree to the parent node  (see Step \ref{step:recurse-high} in Algorithm \ref{alg:rpa}). For notational consistency, we denote each of these ``aggregated" vectors as $\widehat{\mathbf{Y}}_{/\left({\overline{\mathbb{S}}^{(1)},\ldots,\overline{\mathbb{S}}^{(r-3)},\mathbb{S}^{(r-2)}}\right)}$. Via a simple union bound, we see that the probability that each of $\widehat{\mathbf{Y}}_{/\left({\overline{\mathbb{S}}^{(1)},\ldots,\overline{\mathbb{S}}^{(r-3)},\mathbb{S}^{(r-2)}}\right)}$, corresponding to each of the children of the parent node under consideration, equals $\mathbf{0}$, obeys
\begin{align}
	\Pr&\left[\widehat{\mathbf{Y}}_{/\left({\overline{\mathbb{S}}^{(1)},\ldots,\overline{\mathbb{S}}^{(r-3)},\mathbb{S}^{(r-2)}}\right)}= \mathbf{0},\ \text{for all ${\mathbb{S}}^{(r-2)}$} \right]\geq  1-32\left(N^{(r-2)}\right)^4\cdot e^{-\frac{N^{(r-2)}\epsilon^2}{8}}, \label{eq:tempagg}
\end{align}
for any $\epsilon<\eta(p^{(r-2)})$. Thus, defining the event
\[
\mathcal{G}^{(r-2)}:= \left\{\widehat{\mathbf{Y}}_{/\left({\overline{\mathbb{S}}^{(1)},\ldots,\overline{\mathbb{S}}^{(r-3)},\mathbb{S}^{(r-2)}}\right)} = \mathbf{0},\ \text{for all ${\mathbb{S}}^{(r-2)}$}\right\},
\]
we observe that the result in \eqref{eq:tempagg} above is entirely analogous to Proposition \ref{prop:decode2}. 

Similar to the definitions of the ``aggregated" vector  $\widehat{\mathbf{Y}}_{/\left({\overline{\mathbb{S}}^{(1)},\ldots,\overline{\mathbb{S}}^{(r-3)},\mathbb{S}^{(r-2)}}\right)}$, one can define, for every $j\in [r-3]$, the vector  $\widehat{\mathbf{Y}}_{/\left({\overline{\mathbb{S}}^{(1)},\ldots,\overline{\mathbb{S}}^{(j-1)},\mathbb{S}^{(j)}}\right)}$ corresponding to the children of the node indexed by the path $\overline{\mathbb{S}}^{(1)},\ldots,\mathbb{S}^{(j-1)}$. Likewise, for any $j\in [r-3]$, we define the event
\[
\mathcal{G}^{(j)}:= \left\{\widehat{\mathbf{Y}}_{/\left({\overline{\mathbb{S}}^{(1)},\ldots,\overline{\mathbb{S}}^{(j-1)},\mathbb{S}^{(j)}}\right)} = \mathbf{0},\ \text{for all ${\mathbb{S}}^{(j)}$}\right\}.
\]

The following theorem then holds:

\begin{theorem}
	\label{thm:Gr}
	For any $j\in [r-2]$, we have that for any $\epsilon<\eta(p^{(r-2)})$,
	\[
	\Pr\left[\mathcal{G}^{(j)}\right]\geq 1-2^{r+3-j}\cdot \left(N^{(j)}\right)^{r+2-j}\cdot e^{-\frac{N^{(r-2)}\epsilon^2}{8}}.
	\]
\end{theorem}
\begin{proof}
	We proceed by induction on $j$. From \eqref{eq:tempagg}, we see that the statement of the theorem holds for $j=r-2$. Now, assume that the statement holds for some $i<r-2$. Then, arguing similar to the sequence of inequalities \eqref{eq:temp1higher}--\eqref{eq:temp2higher} in the proof of Theorem \ref{thm:agg} (and implicitly using Proposition \ref{prop:flip}), we see that for any subspace $\mathbb{S}^{(i-1)}\subseteq \{0,1\}^{i-1}$, 
	\begin{align*}
		\Pr\left[\widehat{\mathbf{Y}}_{/\left({\overline{\mathbb{S}}^{(1)},\ldots,\overline{\mathbb{S}}^{(i-2)},\mathbb{S}^{(i-1)}}\right)} = \mathbf{0}\right]
		&\geq 1-16\left(N^{(i-1)}\right)^2 \left(N^{(i-1)}-1\right)\cdot e^{-\frac{N^{(i-1)}\epsilon^2}{8}} -\\
		&\ \ \ \ \ \ \ \ \ \ \ \ \ \  \ \ \ \ \  \ 2^{r+3-i}\cdot \left(N^{(i)}\right)^{r+2-i}\cdot e^{-\frac{N^{(r-2)}\epsilon^2}{8}}\\
		&\geq 1-16\left(N^{(i-1)}\right)^3 \cdot e^{-\frac{N^{(i-1)}\epsilon^2}{8}} -\\
		&\ \ \ \ \ \ \ \ \ \ \ \ \ \  \ \ \ \ \  \ 2\cdot \left(N^{(i-1)}\right)^{r+2-i}\cdot e^{-\frac{N^{(r-2)}\epsilon^2}{8}}\\
		&\geq 1-32\cdot \left(N^{(i-1)}\right)^{r+2-i}\cdot e^{-\frac{N^{(r-2)}\epsilon^2}{8}}\\
		&\geq 1-2^{r+4-i}\cdot \left(N^{(i-1)}\right)^{r+2-i}\cdot e^{-\frac{N^{(r-2)}\epsilon^2}{8}}
	\end{align*}
where the third inequality makes use of the fact that $i<r-2\leq m-2$, implying that $N^{(i)}> 4$. Therefore, by a union bound,
\begin{align*}
	\Pr\left[\mathcal{G}^{(i-1)}\right]&= \Pr\left[\widehat{\mathbf{Y}}_{/\left({\overline{\mathbb{S}}^{(1)},\ldots,\overline{\mathbb{S}}^{(i-2)},\mathbb{S}^{(i-1)}}\right)} = \mathbf{0},\ \text{for all $\mathbb{S}^{(i-1)}$}\right]\\
	&\geq 1-2^{r+4-i}\cdot \left(N^{(i-1)}\right)^{r+3-i}\cdot e^{-\frac{N^{(r-2)}\epsilon^2}{8}},
\end{align*}
thereby proving the inductive step.
\end{proof}
The theorem above immediately results in a proof of Theorem \ref{thm:rpaerror}:
%\begin{corollary}
%	\label{cor:rpaerror}
%	For any $\epsilon<\min\{\eta(p^{(r-2)}),\beta(p^{(r-2)})\}$, we have
%	\[
%	P_\text{\normalfont err}(\text{\normalfont RM}(m,r))\leq N^{r+1}\cdot e^{-\frac{N^{(r-2)}\epsilon^2}{4}}.
%	\]
%\end{corollary}
\begin{proof}[Proof of Theorem \ref{thm:rpaerror}]
	From Theorem \ref{thm:Gr}, we see that for any $\epsilon<\eta(p^{(r-2)}) = \eta(\overline{p})$,
	\begin{align*}
		\Pr\left[\mathcal{G}^{(1)}\right]&= \Pr\left[\widehat{\mathbf{Y}}_{/{{\mathbb{S}}^{(1)}}} = \mathbf{0},\ \text{for all $\mathbb{S}^{(1)}$}\right]\\
		&\geq 1-2^{r+2}\cdot (N/2)^{r+1}\cdot e^{-\frac{N^{(r-2)}\epsilon^2}{8}}\\
		&= 1-2 N^{r+1}\cdot e^{-\frac{N^{(r-2)}\epsilon^2}{8}}.
	\end{align*}
Thus, we obtain, via analysis similar to that in the set of inequalities \eqref{eq:temp1higher}--\eqref{eq:temp2higher} in the proof of Theorem \ref{thm:agg}, that
\begin{align*}
P_\text{\normalfont err}(\text{\normalfont RM}(m,r))
&\leq 16N^2(N-1)\cdot e^{-\frac{N\epsilon^2}{8}}+2N^{r+1}\cdot e^{-\frac{N^{(r-2)}\epsilon^2}{8}}\\
&\leq 32N^{r+1}\cdot e^{-\frac{N^{(r-2)}\epsilon^2}{8}} = 32N^{r+1}\cdot e^{-2^{-(r+1)}\cdot{N\epsilon^2}},
\end{align*}
thereby proving our main result.
\end{proof}
%We end this section with a remark. While our analysis of the performance of RPA decoding has restricted attention to projections via one-dimensional subspaces, it will be interesting to study precise estimates of the probabilities of error under an RPA decoding procedure that makes use of higher dimensional subspaces. Based on analysis done in this paper, we believe that the use of higher dimensional subspaces for projection will indeed lead to tighter bounds on $P_\text{\normalfont err}(\text{\normalfont RM}(m,r))$, as against that in Theorem \ref{thm:rpaerror}, even when we restrict attention to the event that the RPA decoder converges after a single iteration. Our intuition for why this may be true is that the union bound arguments made in our analysis appear to scale the (exponentially small) probability of error of decoding a first-order RM code by a factor that roughly equals the total number of edges in the projection-aggregation tree. If one were to use higher dimensional subspaces for projection, the total number of edges in the corresponding projection-aggregation tree can be made significantly smaller. We leave a precise analysis of the probabilities of error in such cases for future work.
\section{Projections Via Higher-Dimensional Subspaces}
\label{sec:higher-subspaces}
In this section, we extend the arguments presented in the previous sections, to the setting where subspaces $\mathbb{B}_i\subseteq \{0,1\}^m$ are now of dimension $k>1$. Such an extension calls for straightfoward modifications of Algorithms \ref{alg:rpa} and \ref{alg:agg}. Furthermore, while the analysis of the FHT decoding step remains unchanged, some important changes need to be made in the analysis of the aggregation step -- in particular, to the convergence rate in Proposition \ref{prop:flip} -- leading to changes in the proof of the rate of convergence of the probability of error $P_\text{err}(\text{RM}(m,r))$. However, since overall structure of the analysis remains unchanged, we present only the details of the modifications here.

%Let a received vector $\mathbf{Y}$ be obtained by passing the $\mathbf{0}$ codeword through a BSC$(p)$.  As usual, we let
%\[
%\mathbf{Y}_{/\mathbb{B}_i}:=\left(\mathbf{Y}_{/\mathbb{B}_i}(T):\ T\in \{0,1\}^m/\mathbb{B}_i\right).
%\]
%The definitions above give rise to the following modifications of the recursive RPA decoding and aggregation algorithms, shown below as Algorithms \ref{alg:rpa-high} and \ref{alg:agg-high}, respectively.

Recall that $\mathbf{Y}\in \{0,1\}^N$ denotes the (random) vector received by passing the input codeword $\mathbf{0}\in \text{RM}(m,r)$ through the BSC$(p)$, and that $\mathbf{Y}\sim \text{Ber}^{\otimes{N}}(p)$.  Our interest, as before, is in characterizing an upper bound on $P_\text{err}(\text{RM}(m,r))$, via a proof of Theorem \ref{thm:rpa-error-high}. To this end, we first present an analogue of our construction of the projection-aggregation tree in Section \ref{sec:higher}.
%, defined the same way as earlier, but now using the \textsc{RPA-High} decoder in Algorithm \ref{alg:rpa-high}

%Theorem \ref{thm:rpa-error-high} is proved via similar arguments as carried out previously. 

For any $t\in \left[\frac{r-1}{k}\right]$, we let $\mathbb{S}^{(1)},\ldots,\mathbb{S}^{(t)}$ be a sequence of $k$-dimensional subspaces, where $\mathbb{S}_j\subseteq \{0,1\}^{m-k(j-1)}$, for each $j\in [t]$. As before, for every $1\leq j\leq t-1$, we index the coordinates of $\mathbf{Y}_{/\left(\mathbb{S}^{(1)},\mathbb{S}^{(2)},\ldots,\mathbb{S}^{(j)}\right)}$ by vectors $\mathbf{z}\in \{0,1\}^{m-kj}$ in lexicographic order, and set
\begin{align*}
	\mathbf{Y}_{/\left(\mathbb{S}^{(1)},\ldots,\mathbb{S}^{(j+1)}\right)}:=\bigg(&\left({Y}_{/\left(\mathbb{S}^{(1)},\mathbb{S}^{(2)},\ldots,\mathbb{S}^{(j)}\right)}\right)_{/\mathbb{S}^{(j+1)}}(T):   T\in \{0,1\}^{m-kj}/\mathbb{S}^{(j+1)}\bigg).
\end{align*}
With the notation in place, it is straightforward to construct a projection-aggregation tree as before, with the root as RM$(m,r)$ and nodes at every level $0\leq i\leq \frac{r-k-1}{k}$ having $n_{k,m-ki}$ children, each of which is an RM$(m-k(i+1),r-k(i+1))$ code. A pictorial depiction of this tree is shown in Figure \ref{fig:projtree-high}. Note that unlike earlier, the leaves of the tree now are first-order RM codes.
\begin{figure*}
	\centering
	\includegraphics[width = 0.7\textwidth]{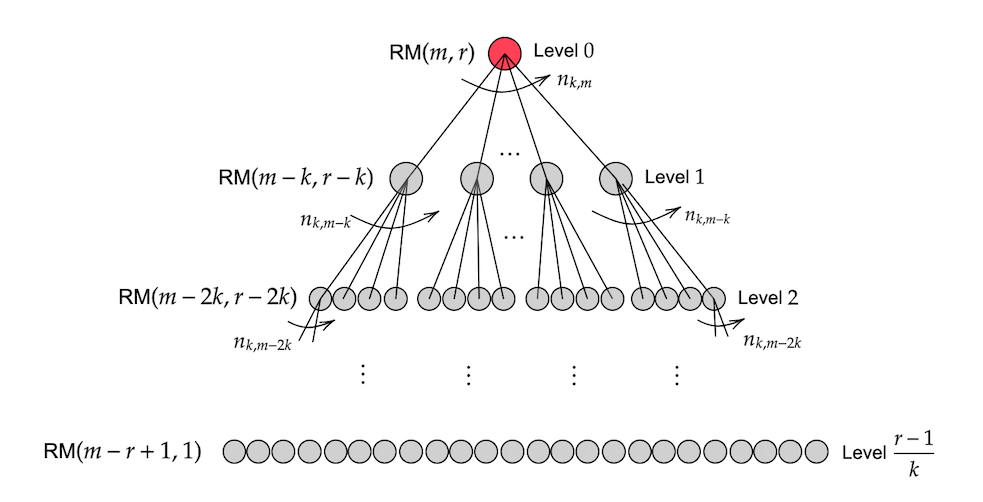}
	\caption{Figure representing a projection-aggregation tree using subspaces of dimension $k$.}
	\label{fig:projtree-high}
\end{figure*}
Now, for any $j\in [\frac{r-1}{k}]$, let $p^{(j)}_k:= \frac12\cdot (1-(1-2p)^{2^{kj}})$ and $N^{(j)}_k:=N/2^{kj}$. Note that we have $p^{(j)} = p_1^{(j)}$ and $N^{(j)} = N_1^{(j)}$, for any $j\in [r-1]$. The following analogue of Lemma \ref{lem:ptree} then holds:
\begin{lemma}
	\label{lem:ptree-high}
	For any choice of $k$-dimensional subspaces $\mathbb{S}^{(1)},\ldots,\mathbb{S}^{(j)}$, for $j\in [\frac{r-1}{k}]$, where $\mathbb{S}_i\subseteq \{0,1\}^{m-k(i-1)}$, for $i\in [j]$, we have
	\[
	\mathbf{Y}_{/\left(\mathbb{S}^{(1)},\ldots,\mathbb{S}^{(j)}\right)}\sim \text{\normalfont Ber}^{\otimes (N_k^{(j)})}(p_k^{(j)}).
	\]
\end{lemma}
Consider any leaf node in the projection-aggregation tree, corresponding to a first-order RM code RM$(m-r+1,1)$, obtained via some fixed path of subspaces $\overline{\mathbb{S}}^{(1)},\ldots,\overline{\mathbb{S}}^{({\frac{r-1}{k}})}$; the projection of the received vector at this node hence is $\mathbf{Y}_{/\left(\overline{\mathbb{S}}^{(1)},\ldots,\overline{\mathbb{S}}^{(\frac{r-1}{k})}\right)}$. Via analysis entirely analogous to Theorem \ref{thm:decode1}, the following lemma holds. Let $\overline{\eta}(p):=\frac{1}{2}\cdot (1-2p)$; note that $\eta(p) = \overline{\eta}(p^{(2)})$.
\begin{lemma}
	\label{lem:decode1-higher}
	We have that for any $\epsilon<\overline{\eta}\left(p^{({r-1})}\right)$,
	\[
	\Pr\Big[\widehat{\mathbf{Y}}_{/\left(\overline{\mathbb{S}}^{(1)},\ldots, \overline{\mathbb{S}}^{(\frac{r-1}{k})}\right)} = \mathbf{0}\Big] \geq 1-16N_k^{(\frac{r-1}{k})}\cdot \text{\normalfont exp}\left({-\frac{N_k^{(\frac{r-1}{k})}\epsilon^2}{4}}\right).
	\]
\end{lemma}
Now, consider the (unique) parent of the leaf node indexed by the path of subspaces $\overline{\mathbb{S}}^{(1)},\ldots,\overline{\mathbb{S}}^{({\frac{r-1}{k}})}$. 
Let $\textsc{Agg}^{(\frac{r-1}{k})}\left(\overline{\mathbb{S}}^{(1)},\ldots,\overline{\mathbb{S}}^{(\frac{r-1}{k})}\right)$ denote the vector $\overline{\mathbf{Y}}$ obtained at the parent node of interest, post aggregation using Algorithm \ref{alg:agg}.

\subsection{Analysis of the Aggregation Step From the Leaves to Their Parents}
We now highlight some of the technical modifications to the existing analysis that need to be made for the aggregation step. Let $\mathcal{G}_k$ denote the event $\left\{\widehat{\mathbf{Y}}_{/\left(\overline{\mathbb{S}}^{(1)},\ldots,\mathbb{S}^{(\frac{r-1}{k})}\right)} = \mathbf{0},\text{\normalfont for all } \mathbb{S}^{(\frac{r-1}{k})}\right\}$. Analogous to Proposition \ref{prop:decode2}, we can argue via Lemma \ref{lem:decode1-higher} that
\begin{equation}
	\label{eq:newdeltaprime}
	\Pr[\mathcal{G}_k]\geq 1-16n_{k,m}N_k^{(\frac{r-1}{k})}\cdot \text{\normalfont exp}\left({-\frac{2N_k^{(\frac{r-1}{k})}\epsilon^2}{8}}\right).
\end{equation}
	Here, we have additionally used the fact that $n_{k,m_1}< n_{k,m_2}$, for $m_1<m_2$, where $n_{k,m}$ is a Gaussian binomial coefficient (see \eqref{eq:gaussbinom}). We wish to analyze the performance of the aggregation step, conditioned on the event $\mathcal{G}_k$. For notational ease, we use the notation $\tilde{\mathbf{Y}}$ to denote ${\mathbf{Y}}_{/\left(\overline{\mathbb{S}}^{(1)},\ldots,\overline{\mathbb{S}}^{(\frac{r-k-1}{k})}\right)}$, with $\widehat{\tilde{\mathbf{Y}}}_{/\mathbb{S}^{(\frac{r-1}{k})}}:=\widehat{\mathbf{Y}}_{/\left(\overline{\mathbb{S}}^{(1)},\ldots,\mathbb{S}^{(\frac{r-1}{k})}\right)}$. Further, we recall that $p_k^{(\frac{r-k-1}{k})} = p^{(r-k-1)}$ and that $N_k^{(\frac{r-k-1}{k})} = N^{(r-k-1)}$. We define $\tilde{n}:=n_{k,m-r+1+k}$. Similar to the case earlier, let $$\phi(\mathbf{x}):= \sum_{i=1}^{\tilde{n}} \mathds{1}\{\tilde{{Y}}_{/\mathbb{B}_i}([\mathbf{x}+\mathbb{B}_i])\neq \widehat{\tilde{{Y}}}_{/\mathbb{B}_i}([\mathbf{x}+\mathbb{B}_i])\}.$$ Upon conditioning on $\mathcal{G}_k$, it can be checked that
\[
\phi(\mathbf{x}) =  \sum_{i=1}^{\tilde{n}} \bigoplus_{\mathbf{b}\in \mathbb{B}_i} \tilde{Y}_{\mathbf{x}\oplus\mathbf{b}},
\]
where the summation over $i\in [n_{k,m}]$ above is over the reals. Note that 
\begin{align}
	\label{eq:newphi}
	\phi(\mathbf{x}) = \begin{cases}
		\sum_{i=1}^{\tilde{n}} \bigoplus\limits_{\mathbf{b}\in [\mathbf{x}+\mathbb{B}_i], \mathbf{b}\neq \mathbf{x}} \tilde{Y}_{\mathbf{b}},\ \text{if $\tilde{Y}_\mathbf{x} = 0$},\\
	 \tilde{n}-\sum_{i=1}^{\tilde{n}} \bigoplus\limits_{\mathbf{b}\in [\mathbf{x}+\mathbb{B}_i], \mathbf{b}\neq \mathbf{x}} \tilde{Y}_{\mathbf{b}},\ \text{if $\tilde{Y}_\mathbf{x} = 1$}.
	\end{cases}
\end{align}
Now, observe that for any $i\in [\tilde{n}]$, each random variable $Z_i:=\bigoplus\limits_{\mathbf{b}\in [\mathbf{x}+\mathbb{B}_i], \mathbf{b}\neq \mathbf{x}} \tilde{Y}_{\mathbf{b}}\sim \text{Ber}(\widehat{p})$, where $$\widehat{p}:=\frac{1}{2}\cdot \left(1-(1-2p^{(r-k-1)})^{2^k-1}\right).$$ However the $Z_i$ random variables are not i.i.d. across $i\in [\tilde{n}]$, unlike in the analysis in Section \ref{sec:agg}. We define $\overline{\phi}(\mathbf{x}) = \frac{\phi(\mathbf{x})}{\tilde{n}}$ and $\overline{\phi}_\infty(\mathbf{x}):= \widehat{p}(1-Y_\mathbf{x})+(1-\widehat{p})Y_\mathbf{x}$. In the lemmas that follow, we argue, using a somewhat powerful theorem about the concentration of functions of i.i.d. random variables, that despite the high correlation among the $Z_i$ random variables, we still have concentration of $\overline{\phi}(\mathbf{x})$ around $\overline{\phi}_\infty(\mathbf{x})$ with sufficiently high probability. 

The theorem that we shall make use of is given below. Here, we define a function $f: \{0,1\}^n\to \mathbb{R}$ to be Lipschitz with Lipschitz constant $c_f\geq 0$ if $|f(\mathbf{u})-f(\mathbf{v})|\leq c_f\cdot d(\mathbf{u},\mathbf{v})$, for all $\mathbf{u},\mathbf{v}\in \{0,1\}^n$. Note that the Lipschitz constant $c_f$ can alternatively be defined as
\begin{equation}
	\label{eq:lip}
c_f = \max_{\mathbf{u}\neq \mathbf{v}} \frac{|f(\mathbf{u})-f(\mathbf{v})|}{d(\mathbf{u},\mathbf{v})}.
\end{equation}
\begin{theorem}[\cite{raginsky}, Thm. 3.4.4]
	\label{thm:rag}
	Let $X_1,\ldots,X_n$ be i.i.d. Ber$(q)$ random variables. Then, for every Lipschitz function
$f : \{0, 1\}^n \to \mathbb{R}$ with Lipschitz constant $c_f$, we have for all $\alpha>0$,
\begin{align*}
&\Pr\left[f(X^n)-\E[f(X^n)]> \alpha\right] \leq \text{\normalfont exp}\left(-\ln\left(\frac{1-q}{q}\right)\cdot \frac{\alpha^2}{nc_f^2\cdot(1-2q)}\right).
\end{align*}
\end{theorem}
Our concentration result follows as an application of the theorem above.
\begin{lemma}
	\label{lem:conc}
	For all $\epsilon>0$, we have that
	\begin{align*}
	&\Pr\left[\left\lvert \sum_{i=1}^{\tilde{n}} Z_i - \tilde{n}\widehat{p}\right\rvert > \tilde{n}\epsilon\right]   \leq 2\cdot\text{\normalfont exp}\left(-\ln\left(\frac{1-p^{(r-k-1)}}{p^{(r-k-1)}}\right)\cdot \frac{N^{(r-k-1)}\epsilon^2\cdot 2^{-2k-2}}{(1-2p^{(r-k-1)})}\right).
	\end{align*}
\end{lemma}
\begin{proof}
	The sum $\sum_{i=1}^{\tilde{n}} Z_i$ of interest can be viewed as a function $f:\{0,1\}^{N^{(r-k-1)}}$ of i.i.d. Ber$(p^{(r-k-1)})$ random variables $\{\tilde{Y}_\mathbf{z}\}$. It can easily be checked, by linearity of expectation, that $\E[\sum_{i=1}^{\tilde{n}} Z_i] = \tilde{n}\widehat{p}$. Fix $\alpha = \tilde{n}\epsilon$. Theorem \ref{thm:rag} then gives us an upper bound on the probability of the event $\{Z_i - \tilde{n}\widehat{p}> \tilde{n}\epsilon\}$; furthermore, by symmetry (and the definition of the Lipschitz constant in \eqref{eq:lip}), an identical upper bound holds on the probability of the event $\{Z_i - \tilde{n}\widehat{p}< -\tilde{n}\epsilon\}$. An application of a union bound hence shows us that for $f(\{Y_\mathbf{z}\}) = \sum_{i=1}^{\tilde{n}} Z_i$, we have
	\begin{align}
	&\Pr\left[|f(X^n)-\E[f(X^n)]|> \tilde{n}\epsilon\right]   \leq 2\ \text{\normalfont exp}\left(-\ln\left(\frac{1-p^{(r-k-1)}}{p^{(r-k-1)}}\right)\cdot \frac{\tilde{n}^2\epsilon^2}{N^{(r-k-1)}c_f^2\cdot(1-2p^{(r-k-1)})}\right). \label{eq:conc-help}
	\end{align}
It thus remains to calculate the Lipschitz constant $c_f$. We claim that $c_f = n_{k-1,m-r+k}$. To see this, first note that the number of subspaces $\mathbb{B}_i\subseteq \{0,1\}^{m-r+k+1}$, $i\in [n_{k,m}]$, that contain a fixed vector $\mathbf{z}$ equals $n_{k-1,m-1}$ (see, e.g., \cite[p. 109]{subspace}).

Let $\{u_\mathbf{z}:\ \mathbf{z}\in \{0,1\}^{m-r+k+1}\}$ and $\{v_\mathbf{z}:\ \mathbf{z}\in \{0,1\}^{m-r+k+1}\}$ differ only in the collection of coordinates $\mathscr{Z}:=\{\mathbf{z}_1,\ldots,\mathbf{z}_\ell\}$, for some $\ell\geq 1$. The following set of inequalities then holds.
\begin{align*}
	|f(\{u_\mathbf{z}\} - f(\{v_\mathbf{z}\})| &= |\{i\in [n_{k,m}]:\ \sum_{j\in [\ell]}\mathds{1}_{\mathbf{z}_j}([\mathbf{x}+\mathbb{B}_i])\ \text{is odd}\}|\\
	&\leq |\{i\in [n_{k,m}]:\ \mathds{1}_{\mathbf{z}_j}([\mathbf{x}+\mathbb{B}_i]) = 1,\ \text{for some}\ j\in [\ell]\}|\\
	&\leq \ell\cdot n_{k-1,m-r+k},
\end{align*}
where the last inequality holds by the remark above on the number of subspaces, or equivalently, the number of cosets of the form $[\mathbf{x}+\mathbb{B}_i]$, over $i\in [\tilde{n}]$, containing a fixed vector $\mathbf{z}_j$, followed by a union bound. Hence, we have that $c_f\leq n_{k-1,m-1}$. Since all inequalities above are met with equality when $\ell=1$, we get that $c_f = n_{k-1,m-1}$.

Plugging this back into \eqref{eq:conc-help} and observing that $$\frac{\tilde{n}}{n_{k-1,m-r+k}} = \frac{n_{k,m-r+1+k}}{n_{k-1,m-r+k}} = \frac{N^{(r-k-1)}-1}{2^k-1}\geq \frac{N^{(r-k-1)}}{2^{k+1}},$$ for $N\geq 2$, gives us the required statement.
\end{proof}
Now, let $\delta':=16n_{k,m}N_k^{(\frac{r-1}{k})}\cdot \text{\normalfont exp}\left({-\frac{2N_k^{(\frac{r-1}{k})}\epsilon^2}{8}}\right)$; note then that
\[
\delta':=2^{4-k}\cdot n_{k,m}N^{(r-k-1)}\cdot \text{\normalfont exp}\left({-\frac{2N^{(r-k-1)}\epsilon^2}{2^{k+3}}}\right),
\]
using our simplified notation. Recall from \eqref{eq:newdeltaprime} that $\Pr[\mathcal{G}_k]\geq 1-\delta'$. Further, let $$\delta:= 2\cdot\text{\normalfont exp}\left(-\ln\left(\frac{1-p^{(r-k-1)}}{p^{(r-k-1)}}\right)\cdot \frac{N^{(r-k-1)}\epsilon^2\cdot 2^{-2k-2}}{(1-2p^{(r-k-1)})}\right).$$ Using Lemma \ref{lem:conc} and arguing similar to Lemma \ref{lem:flip}, we obtain that the event $\overline{\mathcal{G}}_k:=\{|\overline{\phi}(\mathbf{x}) - \overline{\phi}_\infty(\mathbf{x})|\leq \epsilon\}$, for any $\epsilon>0$, is such that
$
\Pr[\overline{\mathcal{G}}_k]\geq 1-\delta,\ \text{for all $\mathbf{x}$}.
$
By arguments entirely analogous to the proofs of Proposition \ref{prop:flip} and Theorem \ref{thm:agg}, we obtain the following theorem. 
%Here, we recall that $p^{(r-2)}:=\frac12\cdot (1-(1-2p)^{2^{r-2}})$.
\begin{theorem}
	\label{thm:agg-high}
	After a single iteration of the RPA decoder at level $\frac{r-1-k}{k}$ of the projection-aggregation tree, we have for all $\epsilon<\overline{\eta}(p^{(r-1)})$ that
	\begin{align*}
	1-\Pr[\overline{\mathbf{Y}} = \mathbf{0}]\leq 32n_{k,m}{N^{(r-k-1)}}^2\cdot  \text{\normalfont exp}\left(-\ln\left(\frac{1-p^{(r-k-1)}}{p^{(r-k-1)}}\right)\cdot 2^{-2k-2}N^{(r-k-1)}\epsilon^2\right).
	\end{align*}
\end{theorem}
\begin{proof}
	The proof follows by noting that via arguments similar to Proposition \ref{prop:flip} and Theorem \ref{thm:agg}, we have $\Pr[\overline{\mathbf{Y}} = \mathbf{0}]\geq 1-N^{(r-k-1)}\cdot(\delta+\delta')-\delta'.$ We bound $\delta+\delta'$ from above by $16n_{k,m}N^{(r-k-1)}\cdot \text{\normalfont exp}\left(-\ln\left(\frac{1-p^{(r-k-1)}}{p^{(r-k-1)}}\right)\cdot 2^{-2k-2}N^{(r-k-1)}\epsilon^2\right)$, via straightforward algebraic manipulations, using the fact that $k>1$, for us.
\end{proof}
Theorem \ref{thm:agg-high} allows us to obtain an upper bound on the probability of error in any aggregation step, conditioned on all projections being decoded correctly, in a manner similar to how Theorem \ref{thm:Gr} uses Corollary \ref{cor:agg}. These estimates can be recursed on, to obtain an upper bound on the probability of error of the RPA decoder in Algorithm \ref{alg:rpa}.
% the decoded estimates of the projections $$\left(\mathbf{Y}_{/\left(\overline{\mathbb{S}}^{(1)},\ldots,\overline{\mathbb{S}}^{(r-2)},\mathbb{S}^{(r-1)}\right)}:\ \mathbb{S}^{(r-1)}\in \{0,1\}^{m-r+2}\right)$$
%Recall that the projections above were decoded using the ML decoder for first-order RM codes. Using the fact (from Lemma \ref{lem:ptree}) that $\mathbf{Y}_{/\left(\overline{\mathbb{S}}^{(1)},\ldots,\overline{\mathbb{S}}^{(r-2)}\right)}\sim \text{Ber}^{\otimes(N^{(r-2)})}(p^{(r-2)})$, we obtain via our analysis of the error probabilities for second-order RM codes (in particular, Corollary \ref{cor:agg}) that for any $\epsilon<\eta(p^{(r-2)})$,
\subsection{Recursive Analysis of Error Probabilities}
In this subsection, we shall recursively analyze the probabilities of error incurred in aggregation steps at different levels of the projection-aggregation tree, similar to Theorem \ref{thm:Gr}. For any $j\in [\frac{r-k-1}{k}]$, we define the event
\[
\mathcal{G}^{(j)}_k:= \left\{\widehat{\mathbf{Y}}_{/\left({\overline{\mathbb{S}}^{(1)},\ldots,\overline{\mathbb{S}}^{(j-1)},\mathbb{S}^{(j)}}\right)} = \mathbf{0},\ \text{for all ${\mathbb{S}}^{(j)}$}\right\}.
\]
The following theorem then holds:
\begin{theorem}
	\label{thm:Gr-high}
	For any $j\in [\frac{r-k-1}{k}]$, we have that for any $\epsilon<\overline{\eta}(p^{(r-1)})$,
	\begin{align*}
	\Pr\left[\mathcal{G}^{(j)}_k\right]
	&\geq 1- 2^{\frac{k(4-j)+r-1}{k}}\cdot n_{k,m}^{\frac{r-1-kj}{k}} \left(N_k^{(j)}\right)^{2}\cdot  \text{\normalfont exp}\left({-\ln\left(\frac{1-p^{(r-k-1)}}{p^{(r-k-1)}}\right)\cdot 2^{-2k-2}N^{({r-k-1})}\epsilon^2}\right).
	\end{align*}
\end{theorem}
\begin{proof}
	The proof proceeds by induction on $j$. For the base case, we consider $j = \frac{r-1-k}{k}$. From Theorem \ref{thm:agg-high} and a union bound, we verify the correctness of the lower bound in the theorem statement. Now, suppose that the lower bound holds for some $i<\frac{r-1-k}{k}$. Via arguments similar to that in Theorem \ref{thm:Gr}, we see that
	\begin{align*}
		1-\Pr[\mathcal{G}^{(i-1)}_k]
		&\leq N_k^{(i-1)}\cdot \delta+(1-\Pr[\mathcal{G}_k^{(i)}])\\
		&\leq 16n_{k,m}\left(N_k^{(i-1)}\right)^3 \text{\normalfont exp}\left(-\ln\left(\frac{1-p^{(r-k-1)}}{p^{(r-k-1)}}\right)\cdot 2^{-2k-2}N_k^{(i-1)}\epsilon^2\right)\\
		&\ \ \ \ \ \ \  \  + 2^{\frac{k(4-i)+r-1}{k}}\cdot n_{k,m}^{\frac{r-1-ki}{k}} \left(N_k^{(i)}\right)^{2}\cdot \text{\normalfont exp}\left({-\ln\left(\frac{1-p^{(r-k-1)}}{p^{(r-k-1)}}\right)\cdot 2^{-2k-2}N^{({r-k-1})}\epsilon^2}\right)\\
		&\leq 2^{\frac{k(4-i+1)+r-1}{k}}\cdot n_{k,m}^{\frac{r-1-k(i-1)}{k}}\left(N_k^{(i)}\right)^{3}\cdot  \text{\normalfont exp}\left({-\ln\left(\frac{1-p^{(r-k-1)}}{p^{(r-k-1)}}\right)\cdot 2^{-2k-2}N^{({r-k-1})}\epsilon^2}\right),
	\end{align*}
which completes the proof.
\end{proof}
As a corollary, we obtain that
\begin{align}
	1-\Pr\left[\mathcal{G}^{(1)}_k\right]
	&\leq 2^{\frac{r-1+3k}{k}-2k}\cdot n_{k,m}^{\frac{r-1}{k}} N^{2}\cdot \text{\normalfont exp}\left({-\ln\left(\frac{1-p^{(r-k-1)}}{p^{(r-k-1)}}\right)\cdot 2^{-1-r-k}N\epsilon^2}\right),\label{eq:G1-high}
\end{align}
where we have used the fact that $N^{({r-k-1})} = N\cdot 2^{-r+k+1}$. This leads to the following theorem, which mirrors Theorem \ref{thm:rpaerror}.

\begin{proof}[Proof of Thm. \ref{thm:rpa-error-high}]
	Via arguments similar to that in the proof of Theorem \ref{thm:Gr-high}, we see that
	\begin{align}
		P_\text{\normalfont err}(\text{\normalfont RM}(m,r))
		&\leq 16n_{k,m}N^3\cdot \text{\normalfont exp}\left(-\ln\left(\frac{1-p^{(r-k-1)}}{p^{(r-k-1)}}\right)\cdot 2^{-2k-2}N\epsilon^2\right) + \notag\\
		&\ \ \ \ \ \ \  \ 2^{\frac{r-1+3k}{k}-2k}\cdot n_{k,m}^{\frac{r-1}{k}} N^{2}\cdot \text{\normalfont exp}\left({-\ln\left(\frac{1-p^{(r-k-1)}}{p^{(r-k-1)}}\right)\cdot 2^{-1-r-k}N\epsilon^2}\right) \notag\\
		&\leq 64N^3\cdot n_{k,m}^{\frac{r-1}{k}}\cdot \text{\normalfont exp}\left(-\ln\left(\frac{1-p^{(r-k-1)}}{p^{(r-k-1)}}\right)\cdot 2^{-r-1-k}N\epsilon^2\right).
		\label{eq:temp-higher}
%		\\
%		&\leq 64N^3\cdot n_{k,m}^{\frac{r-1}{k}}\cdot \text{\normalfont exp}\left(-\ln\left(\frac{1-p^{(r-2)}}{p^{(r-2)}}\right)\cdot 2^{-r-3}N\epsilon^2\right)
	\end{align}
where the first inequality uses \eqref{eq:G1-high} and the second inequality uses the fact that $2k+2\leq r+1+k$, since $k|(r-1)$, thereby proving the theorem, using the observation that $\overline{\eta}(p^{(r-1)}) = \eta(p^{(r-2)}) = \eta(\overline{p})$.
% and the third uses the fact that $k\geq 2$. 
%(\textbf{The last inequality here may be unnecessary; possibly can be optimized \ldots})
\end{proof}
Given the upper bound on $P_\text{\normalfont err}(\text{\normalfont RM}(m,r))$ in Theorem \ref{thm:rpa-error-high}, we proceed to argue that a good choice of the dimension $k$ of subspaces used for projection, in order to arrive at the best upper bound is $k=r-1$. Indeed, the following result holds.
\begin{corollary}
	\label{cor:error-high}
	We have that for any $\epsilon < \eta(p^{(r-2)})$,
	\[
	P_\text{\normalfont err}(\text{\normalfont RM}(m,r))\leq N^3\cdot 2^{(m-r+2)(r-1)+6}\cdot\  \text{\normalfont exp}\left(-\ln\left(\frac{1-p}{p}\right)\cdot N\epsilon^2\right).
	\]
	Hence, for any fixed $\delta\in (0,1)$, we have that $P_\text{\normalfont err}(\text{\normalfont RM}(m,r))\leq 2^{\overline{\rho}_m(r,\delta)}$, where
	\begin{align*}
		&\overline{\rho}_m(r,\delta):=
		(m-r+2)(r-1)-\frac{\delta\log_2(e)}{16}\cdot \ln\left(\frac{1-p}{p}\right)\cdot 2^m\cdot(1-2p)^{2^{r+1}}+3m+6.
	\end{align*}
\end{corollary}
\begin{proof}
	Consider the term
	\[
	\theta(k):= -\ln\left(\frac{1-p^{(r-k-1)}}{p^{(r-k-1)}}\right)\cdot 2^{-r-1-k}
	\]
	that appears in the exponent in \eqref{eq:temp-higher} in the proof of Theorem \ref{thm:rpa-error-high}. Since $p^{(r-k-1)}$ is decreasing in $k\in [r-1]$, we have that $-\ln\left(\frac{1-p^{(r-k-1)}}{p^{(r-k-1)}}\right)$ is also decreasing in $k\in [r-1]$. Likewise, the expression $2^{-r-1-k}$ is decreasing in $k\in [r-1]$; putting everything together, we see that $\theta(k)$ is decreasing with $k\in [r-1]$, thereby showing that the exponent in the upper bound on $P_\text{\normalfont err}(\text{\normalfont RM}(m,r))$ in \eqref{eq:temp-higher} is minimized by setting $k = r-1$. 
	
	We next derive an upper bound on the multiplicative prefactor $n_{k,m}^{\frac{r-1}{k}}$ in \eqref{eq:temp-higher}. From \eqref{eq:gaussbinom}, we have that
	\begin{equation}
		\label{eq:surr}
		n_{k,m}^{\frac{r-1}{k}}\leq \left((2^{m-k+1})^k\right)^{\frac{r-1}{k}} = 2^{(m-k+1)(r-1)}.
	\end{equation}
	Clearly, the upper bound in \eqref{eq:surr} is strictly decreasing in $k\in [r-1]$. The above arguments thus give rise to the first statement of the corollary.
	
	Straightforward algebraic manipulations then give us the second statement, via arguments similar to that in Corollary \ref{cor:logr}, by writing $\epsilon^2 = \delta\cdot \eta(\overline{p})^2$, for some $\delta\in (0,1)$.
%	\begin{align}
%		P_\text{\normalfont err}(\text{\normalfont RM}(m,r)) \leq 64N^3\cdot n_{k,m}^{\frac{r-1}{k}}\cdot \text{\normalfont exp}\left(-\ln\left(\frac{1-p}{p}\right)\cdot N\epsilon^2\right),
%	\end{align}
%	for any $\epsilon<\overline{\eta}(p^{(r-1)}) = \eta(\overline{p})$, where $\overline{p} = p^{(r-2)}$, thereby proving Theorem \ref{thm:rpa-error-high}.
\end{proof}

%Observe that the only explicit  First, observe that the only term in the upper bound above that contains an explicit dependence on the dimension $k$ of the subspaces used, is the 
%Following on from Corollary 
%The following corollary then holds.
%\begin{corollary}
%	\label{cor:error-high}
%	For any fixed $\delta\in (0,1)$, we have that $P_\text{\normalfont err}(\text{\normalfont RM}(m,r))\leq 2^{\overline{\rho}_m(r,\delta)}$, where
%	\begin{align*}
%	&\overline{\rho}_m(r,\delta):=\\
%	&(m-r+2)(r-1)-\frac{\delta\log_2(e)}{16}\cdot \ln\left(\frac{1-p}{p}\right)\cdot 2^m\cdot(1-2p)^{2^{r+1}}\\
%	&\ \ \ \ \ \ \ \ \ \ \ \ \ \ \ \ \ \ \ \ \ \ \ \ \ \ \ +3m+6.
%	\end{align*}
%\end{corollary}
%\begin{proof}
%	The proof follows by using $k=r-1$ in the upper bound expression in Theorem \ref{thm:rpa-error-high}. 
%%	To this end, we use \eqref{eq:surr} and the fact that
%%	\[
%%	\ln\left(\frac{1-p^{(r-2)}}{p^{(r-2)}}\right)\geq \ln\left(\text{exp}\left(2\cdot (1-2p)^{2^{r-2}}\right)\right) = 2\cdot (1-2p)^{2^{r-2}}.
%%	\]
%	Straightforward algebraic manipulations then give us the required claim, via arguments similar to that in Corollary \ref{cor:logr}, by writing $\epsilon^2 = \delta\cdot \eta(\overline{p})^2$, for some $\delta\in (0,1)$.
%\end{proof}
The expression $\overline{\rho}_m(r,\delta)$ in Corollary \ref{cor:error-high} is very similar to the expression $\rho_m(r,\delta)$ in the proof of Corollary \ref{cor:logr}, with $P_\text{\normalfont err}(\text{\normalfont RM}(m,r))\leq 2^{{\rho}_m(r,\delta)}$, using one-dimensional subspaces. We hence recover the exact same growth rate of $r$ with $m$ to guarantee vanishing error probabilities for RPA decoding with higher-dimensional subspaces, as we did in Corollary \ref{cor:logr}.
\section{Conclusion and Future Work}
\label{sec:conclusion}
In this paper, we provided a formal, theoretical analysis of the performance of the Recursive Projection-Aggregation (RPA) decoder for Reed-Muller (RM) codes over a binary symmetric channel (BSC). Via our analysis, we obtained that the RPA decoder recovers the correct input codeword with overwhelming probability, for RM codes of order logarithmic in $m$, when the blocklength $N=2^m$ is sufficiently large. Key components of our analysis were our estimates of the error probabilities of the Fast Hadamard Transform-based decoder and the aggregation step in the RPA decoder.

An interesting line of future work would be to extend the analysis in this paper to situations where more than one iteration of the RPA decoder is run, before convergence. We expect that such an attempt will allow one to obtain vanishing error probabilities for higher RM orders than guaranteed in this paper. Another important extension of our analysis will be to the setting of the transmission of RM codewords over general binary-input, memoryless, symmetric (BIMS) channels (or even over general discrete memoryless channels or DMCs), of which the BSC is a special case.

\section*{Acknowledgement}
This work was carried out in part when the second author was visiting the Simons Institute for the Theory of Computing, Berkeley, during Jan-May 2024 and the Indian Institute of Science, Bengaluru, during Aug-Dec 2024. The authors thank Ilya Dumer, Navin Kashyap, and Henry Pfister for helpful comments and pointers to the literature.
\ifCLASSOPTIONcaptionsoff
  \newpage
\fi



%\bibliographystyle{IEEEtran}

\bibliographystyle{IEEEtran}
{\footnotesize
	\bibliography{references}}

\begin{thebibliography}{10}
\providecommand{\url}[1]{#1}
\csname url@samestyle\endcsname
\providecommand{\newblock}{\relax}
\providecommand{\bibinfo}[2]{#2}
\providecommand{\BIBentrySTDinterwordspacing}{\spaceskip=0pt\relax}
\providecommand{\BIBentryALTinterwordstretchfactor}{4}
\providecommand{\BIBentryALTinterwordspacing}{\spaceskip=\fontdimen2\font plus
\BIBentryALTinterwordstretchfactor\fontdimen3\font minus
  \fontdimen4\font\relax}
\providecommand{\BIBforeignlanguage}[2]{{%
\expandafter\ifx\csname l@#1\endcsname\relax
\typeout{** WARNING: IEEEtran.bst: No hyphenation pattern has been}%
\typeout{** loaded for the language `#1'. Using the pattern for}%
\typeout{** the default language instead.}%
\else
\language=\csname l@#1\endcsname
\fi
#2}}
\providecommand{\BIBdecl}{\relax}
\BIBdecl

\bibitem{reed}
I.~Reed, ``A class of multiple-error-correcting codes and the decoding
  scheme,'' \emph{Transactions of the IRE Professional Group on Information
  Theory}, vol.~4, no.~4, pp. 38--49, 1954.

\bibitem{muller}
D.~E. Muller, ``Application of boolean algebra to switching circuit design and
  to error detection,'' \emph{Transactions of the I.R.E. Professional Group on
  Electronic Computers}, vol. EC-3, no.~3, pp. 6--12, 1954.

\bibitem{kud1}
S.~Kudekar, S.~Kumar, M.~Mondelli, H.~D. Pfister, E.~Sasoglu, and R.~L.
  Urbanke, ``Reed-{M}uller codes achieve capacity on erasure channels,''
  \emph{IEEE Transactions on Information Theory}, vol.~63, no.~7, pp.
  4298--4316, 2017.

\bibitem{Reeves}
G.~Reeves and H.~D. Pfister, ``Reed–{Muller} codes on {BMS} channels achieve
  vanishing bit-error probability for all rates below capacity,'' \emph{IEEE
  Transactions on Information Theory}, pp. 1--1, 2023.

\bibitem{abbesandon}
E.~Abbe and C.~Sandon, ``A proof that {{Reed-Muller} codes achieve Shannon}
  capacity on symmetric channels,'' in \emph{2023 IEEE 64th Annual Symposium on
  Foundations of Computer Science (FOCS)}, 2023, pp. 177--193.

\bibitem{arikanrmpolar-1}
E.~Arikan, ``A performance comparison of polar codes and {Reed-Muller} codes,''
  \emph{IEEE Communications Letters}, vol.~12, no.~6, pp. 447--449, 2008.

\bibitem{arikanrmpolar-2}
------, ``A survey of {Reed-Muller} codes from polar coding perspective,'' in
  \emph{2010 IEEE Information Theory Workshop on Information Theory (ITW 2010,
  Cairo)}, 2010, pp. 1--5.

\bibitem{polar}
------, ``Channel polarization: A method for constructing capacity-achieving
  codes for symmetric binary-input memoryless channels,'' \emph{IEEE
  Transactions on Information Theory}, vol.~55, no.~7, pp. 3051--3073, 2009.

\bibitem{shuvaltal}
B.~Shuval and I.~Tal, ``Fast polarization for processes with memory,''
  \emph{IEEE Transactions on Information Theory}, vol.~65, no.~4, pp.
  2004--2020, 2019.

\bibitem{sasoglutal1}
E.~\c{S}a\c{s}o\v{g}lu and I.~Tal, ``Polar coding for processes with memory,''
  \emph{IEEE Transactions on Information Theory}, vol.~65, no.~4, pp.
  1994--2003, 2019.

\bibitem{fht2}
R.~R. Green, ``A serial orthogonal decoder,'' \emph{JPL Space Programs
  Summary}, vol. 37–39-IV, p. 247–253, 1966.

\bibitem{fht}
Y.~Be'ery and J.~Snyders, ``Optimal soft decision block decoders based on fast
  {Hadamard} transform,'' \emph{IEEE Transactions on Information Theory},
  vol.~32, no.~3, pp. 355--364, 1986.

\bibitem{sidelnikov}
V.~M. Sidel'nikov and A.~S. Pershakov, ``Decoding of {Reed-Muller} codes with a
  large number of errors,'' \emph{Problemy Peredachi Informatsii}, vol.~28,
  no.~3, p. 80–94, 1992.

\bibitem{sakkour}
B.~Sakkour, ``Decoding of second order {Reed-Muller} codes with a large number
  of errors,'' in \emph{IEEE Information Theory Workshop}, 2005, p.~3.

\bibitem{dumer1}
I.~Dumer, ``Recursive decoding and its performance for low-rate {Reed-Muller}
  codes,'' \emph{IEEE Transactions on Information Theory}, vol.~50, no.~5, pp.
  811--823, 2004.

\bibitem{dumer2}
------, ``Soft-decision decoding of {Reed-Muller} codes: a simplified
  algorithm,'' \emph{IEEE Transactions on Information Theory}, vol.~52, no.~3,
  pp. 954--963, 2006.

\bibitem{dumer3}
I.~Dumer and K.~Shabunov, ``Soft-decision decoding of {Reed-Muller} codes:
  recursive lists,'' \emph{IEEE Transactions on Information Theory}, vol.~52,
  no.~3, pp. 1260--1266, 2006.

\bibitem{dumer_burnashev}
M.~Burnashev and I.~Dumer, ``Error exponents for recursive decoding of
  {Reed–Muller} codes on a binary-symmetric channel,'' \emph{IEEE
  Transactions on Information Theory}, vol.~52, no.~11, pp. 4880--4891, 2006.

\bibitem{ko}
\BIBentryALTinterwordspacing
A.~V. Makkuva, X.~Liu, M.~V. Jamali, H.~Mahdavifar, S.~Oh, and P.~Viswanath,
  ``{KO} codes: inventing nonlinear encoding and decoding for reliable wireless
  communication via deep-learning,'' in \emph{Proceedings of the 38th
  International Conference on Machine Learning}, ser. Proceedings of Machine
  Learning Research, M.~Meila and T.~Zhang, Eds., vol. 139.\hskip 1em plus
  0.5em minus 0.4em\relax PMLR, 18--24 Jul 2021, pp. 7368--7378. [Online].
  Available: \url{https://proceedings.mlr.press/v139/makkuva21a.html}
\BIBentrySTDinterwordspacing

\bibitem{jamali}
M.~V. Jamali, X.~Liu, A.~V. Makkuva, H.~Mahdavifar, S.~Oh, and P.~Viswanath,
  ``Machine learning-aided efficient decoding of {Reed-Muller} subcodes,''
  \emph{IEEE Journal on Selected Areas in Information Theory}, vol.~4, pp.
  260--275, 2023.

\bibitem{redundant1}
E.~Santi, C.~Hager, and H.~D. Pfister, ``Decoding {Reed-Muller} codes using
  minimum- weight parity checks,'' in \emph{2018 IEEE International Symposium
  on Information Theory (ISIT)}, 2018, pp. 1296--1300.

\bibitem{redundant2}
M.~Lian, C.~Häger, and H.~D. Pfister, ``Decoding {Reed-Muller} codes using
  redundant code constraints,'' in \emph{2020 IEEE International Symposium on
  Information Theory (ISIT)}, 2020, pp. 42--47.

\bibitem{rpa}
M.~Ye and E.~Abbe, ``Recursive projection-aggregation decoding of {Reed-Muller}
  codes,'' in \emph{2019 IEEE International Symposium on Information Theory
  (ISIT)}, 2019, pp. 2064--2068.

\bibitem{rpacomplex1}
J.~Li, S.~M. Abbas, T.~Tonnellier, and W.~J. Gross, ``Reduced complexity {RPA}
  decoder for {Reed-Muller} codes,'' in \emph{2021 11th International Symposium
  on Topics in Coding (ISTC)}, 2021, pp. 1--5.

\bibitem{rpacomplex2}
D.~Fathollahi, N.~Farsad, S.~A. Hashemi, and M.~Mondelli, ``Sparse
  multi-decoder recursive projection aggregation for {Reed-Muller} codes,'' in
  \emph{2021 IEEE International Symposium on Information Theory (ISIT)}, 2021,
  pp. 1082--1087.

\bibitem{asw}
E.~Abbe, A.~Shpilka, and A.~Wigderson, ``{Reed–Muller} codes for random
  erasures and errors,'' \emph{IEEE Transactions on Information Theory},
  vol.~61, no.~10, pp. 5229--5252, 2015.

\bibitem{sberlo}
O.~Sberlo and A.~Shpilka, ``On the performance of {Reed-Muller} codes with
  respect to random errors and erasures,'' in \emph{Proceedings of the
  Thirty-First Annual ACM-SIAM Symposium on Discrete Algorithms}, ser. SODA
  '20.\hskip 1em plus 0.5em minus 0.4em\relax USA: Society for Industrial and
  Applied Mathematics, 2020, p. 1357–1376.

\bibitem{mws}
F.~J. MacWilliams and N.~J.~A. Sloane, \emph{The Theory of Error-Correcting
  Codes}, 2nd~ed.\hskip 1em plus 0.5em minus 0.4em\relax North-Holland, 1978.

\bibitem{rm_survey}
E.~Abbe, A.~Shpilka, and M.~Ye, ``Reed-{M}uller codes: Theory and algorithms,''
  \emph{IEEE Transactions on Information Theory}, vol.~67, no.~6, pp.
  3251--3277, 2021.

\bibitem{donnell}
R.~O’Donnell, \emph{Analysis of Boolean Functions}.\hskip 1em plus 0.5em
  minus 0.4em\relax Cambridge University Press, 2014.

\bibitem{hoeffding}
\BIBentryALTinterwordspacing
W.~Hoeffding, ``Probability inequalities for sums of bounded random
  variables,'' \emph{Journal of the American Statistical Association}, vol.~58,
  no. 301, pp. 13--30, 1963. [Online]. Available:
  \url{http://www.jstor.org/stable/2282952}
\BIBentrySTDinterwordspacing

\bibitem{hdp}
R.~Vershynin, \emph{High-Dimensional Probability: An Introduction with
  Applications in Data Science}, ser. Cambridge Series in Statistical and
  Probabilistic Mathematics.\hskip 1em plus 0.5em minus 0.4em\relax Cambridge
  University Press, 2018.

\bibitem{coseterror}
F.~Chen, B.~Zhang, and Q.~Huang, ``Coset error patterns in recursive
  projection-aggregation decoding,'' in \emph{2024 IEEE International Symposium
  on Information Theory (ISIT)}, 2024, pp. 921--926.

\bibitem{raginsky}
M.~Raginsky and I.~Sason, \emph{Concentration of Measure Inequalities in
  Information Theory, Communications, and Coding: Third Edition}.\hskip 1em
  plus 0.5em minus 0.4em\relax Boston-Delft: Now publishers, 2018.

\bibitem{subspace}
M.~Greferath, M.~O. Pavcevic, N.~Silberstein, and M.~A. Vazquez-Castro,
  \emph{Network Coding and Subspace Designs}, 1st~ed.\hskip 1em plus 0.5em
  minus 0.4em\relax Springer Publishing Company, Incorporated, 2018.

\end{thebibliography}

%\clearpage
%\appendix
%\input{appendix.tex}

%\begin{IEEEbiographynophoto}[{\includegraphics[width=1in,height=1.25in,clip,keepaspectratio]{mshell}}]{Michael Shell}

%\begin{IEEEbiographynophoto}{V.~Arvind Rameshwar}
%\begin{IEEEbiographynophoto}{Navin Kashyap}



%\vfill


\end{document}